\documentclass[oneside,reqno]{amsart}
\usepackage{amssymb,amsmath,stmaryrd,txfonts,mathrsfs,amsthm}
\usepackage{arydshln}

\usepackage{color}
\definecolor{myurlcolor}{rgb}{0.5,0,0}
\definecolor{mycitecolor}{rgb}{0,0,1}
\definecolor{myrefcolor}{rgb}{0,0,1}
\usepackage[pagebackref]{hyperref}
\hypersetup{colorlinks,
linkcolor=myrefcolor,
citecolor=mycitecolor,
urlcolor=myurlcolor}

\newcommand{\define}[1]{{\bf \boldmath{#1}}}
\newcommand{\maps}{\colon}

\newcommand{\R}{{\mathbb R}}

\newcommand{\Mark}{\mathsf{Mark}}
\newcommand{\FinSet}{\mathsf{FinSet}}
\newcommand{\FinVect}{\mathsf{FinVect}}
\newcommand{\Dynam}{\mathsf{Dynam}}
\newcommand{\SemiAlgRel}{\mathsf{SemiAlgRel}}
\newcommand{\LinRel}{\mathsf{LinRel}}

\newcommand{\MMark}{\mathbb{M}\mathbf{ark}}
\newcommand{\LLinRel}{\mathbb{L}\mathbf{inRel}}

\newcommand{\lD}{\ensuremath{\mathbb{D}}}

\usepackage[all,2cell]{xy}
\usepackage[neveradjust]{paralist}
\usepackage{mathtools}
\usepackage{multirow}
\usepackage[outline]{contour}
\contourlength{1.2pt}

\usepackage{tikz}
\usepackage{tikz-cd}
\usepackage{xcolor}
\usepackage{framed,color}
\usetikzlibrary{matrix,arrows,decorations.pathmorphing,positioning,shapes}

\tikzstyle{main node} =[circle,fill=white!20,draw,font=\sffamily\Large\bfseries]
\tikzstyle{terminal}=[circle,fill=white!20,draw,font=\sffamily\Large\bfseries,color=purple,fill=none]

\makeatletter
\let\ea\expandafter

\def\mdef#1#2{\ea\ea\ea\gdef\ea\ea\noexpand#1\ea{\ea\ensuremath\ea{#2}}}
\def\alwaysmath#1{\ea\ea\ea\global\ea\ea\ea\let\ea\ea\csname your@#1\endcsname\csname #1\endcsname
  \ea\def\csname #1\endcsname{\ensuremath{\csname your@#1\endcsname}}}

\mdef\fahat{\hat{\fa}}

\def\rightarrowtailfill@{\arrowfill@{\Yright\joinrel\relbar}\relbar\rightarrow}
\newcommand\xrightarrowtail[2][]{\ext@arrow 0055{\rightarrowtailfill@}{#1}{#2}}

\def\twoheadrightarrowfill@{\arrowfill@{\relbar\joinrel\relbar}\relbar\twoheadrightarrow}
\newcommand\xtwoheadrightarrow[2][]{\ext@arrow 0055{\twoheadrightarrowfill@}{#1}{#2}}

\def\slashedarrowfill@#1#2#3#4#5{%
  $\m@th\thickmuskip0mu\medmuskip\thickmuskip\thinmuskip\thickmuskip
   \relax#5#1\mkern-7mu%
   \cleaders\hbox{$#5\mkern-2mu#2\mkern-2mu$}\hfill
   \mathclap{#3}\mathclap{#2}%
   \cleaders\hbox{$#5\mkern-2mu#2\mkern-2mu$}\hfill
   \mkern-7mu#4$%
}
\def\rightslashedarrowfill@{%
  \slashedarrowfill@\relbar\relbar\mapstochar\rightarrow}
\newcommand\xslashedrightarrow[2][]{%
  \ext@arrow 0055{\rightslashedarrowfill@}{#1}{#2}}
\mdef\hto{\xslashedrightarrow{}}
\mdef\htoo{\xslashedrightarrow{\quad}}

\theoremstyle{plain}

\newtheorem{thm}{Theorem}[section]
\newtheorem{lem}[thm]{Lemma}
\newtheorem{prop}[thm]{Proposition}

\newtheorem{con}[thm]{Conjecture}

\newtheorem{defn}[thm]{Definition}

\theoremstyle{remark}

\def\thmqedhere{\expandafter\csname\csname @currenvir\endcsname @qed\endcsname}

\makeatother

\UseAllTwocells

\begin{document}
\sloppy

\title{Coarse-graining open Markov processes}

\maketitle
\begin{center}   
  {\em John\ C.\ Baez \\}
  \vspace{0.3cm}
  {\small
 Department of Mathematics \\
    University of California \\
  Riverside CA, USA 92521 \\ and \\
 Centre for Quantum Technologies  \\
    National University of Singapore \\
    Singapore 117543  \\    } 
  \vspace{0.4cm}
{\em Kenny Courser \\}
\vspace{0.3cm}
   {\small
   Department of Mathematics \\
  University of California \\
  Riverside CA, USA 92521 \\}
  \vspace{0.3cm}   
  {\small email:  baez@math.ucr.edu, courser@math.ucr.edu\\} 
  \vspace{0.3cm}   
  {\small \today}
  \vspace{0.3cm}   
\end{center} 

\begin{abstract}
\noindent
Coarse-graining is a standard method of extracting a simpler Markov process from a more complicated one by identifying states.  Here we extend coarse-graining to `open' Markov processes: that is, those where probability can flow in or out of certain states called `inputs' and `outputs'.  One can build up an ordinary Markov process from smaller open pieces in two basic ways: composition, where we identify the outputs of one open Markov process with the inputs of another, and tensoring, where we set two open Markov processes side by side.  In previous work, Fong, Pollard and the first author showed that these constructions make open Markov processes into the morphisms of a symmetric monoidal category.   Here we go further by constructing a symmetric monoidal \emph{double} category where the 2-morphisms include ways of coarse-graining open Markov processes.    We also extend the already known `black-boxing' functor from the category of open Markov processes to our double category.   Black-boxing sends any open Markov process to the linear relation between input and output data that holds in steady states, including nonequilibrium steady states where there is a nonzero flow of probability through the process.  To extend black-boxing to a functor between double categories, we need to prove that black-boxing is compatible with coarse-graining.  
\end{abstract}

\vfill \eject
\section{Introduction}
\label{sec:intro}

A `Markov process' is a stochastic model describing a sequence of transitions between states in which the probability of a transition depends only on the current state.   The only Markov processes we consider here are continuous-time Markov processes with a finite set of states.   Such a Markov process can be drawn as a labeled graph:
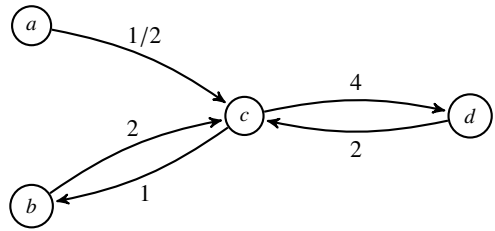
\[
\begin{tikzpicture}[->,>=stealth',shorten >=1pt,thick,scale=1]
  \node[main node] (1) at (0,2.2) {$\scriptstyle{a}$};
  \node[main node](2) at (0,-.2) {$\scriptstyle{b}$};
  \node[main node](3) at (2.83,1)  {$\scriptstyle{c}$};
  \node[main node](4) at (5.83,1) {$\scriptstyle{d}$};
  \path[every node/.style={font=\sffamily\small}, shorten >=1pt]
    (3) edge [bend left=12] node[above] {$4$} (4)
    (4) edge [bend left=12] node[below] {$2$} (3)
    (2) edge [bend left=12] node[above] {$2$} (3) 
    (3) edge [bend left=12] node[below] {$1$} (2)
    (1) edge [bend left=12] node[above] {$1/2$}(3);
\end{tikzpicture}
\]
In this example the set of states is $X = \{a,b,c,d\}$.   The numbers labeling edges
are transition rates, so the probability $\pi_i(t)$ of being in state $i \in X$ at time 
$t \in \R$ evolves according to a linear differential equation
\[    \frac{d}{dt}\, \pi_i(t) = \sum_{j \in X} H_{ij}\, \pi_j(t) \]
called the `master equation', where the matrix $H$ can be read off from the diagram:
\[ H=
\left[\begin{array}{rrrr}
    -1/2    & 0    & 0    & 0  \\
      0                   & -2   & 1   & 0 \\
     1/2   & 2     & -5  & 2 \\
      0                  & 0      & 4   & -2 \\
\end{array}\right].
\]
If there is an edge from a state $j$ to a distinct state $i$, the matrix entry $H_{ij}$ is 
the number labeling that edge, while if there is no such edge, $H_{ij} = 0$.  The diagonal
entries $H_{ii}$ are determined by the requirement that the sum of each column is zero.  This requirement says that the rate at which probability leaves a state equals the rate at which it
goes to other states.   As a consequence, the total probability is conserved:
\[    \frac{d}{dt} \sum_{i \in X} \pi_i(t) = 0 \]
and is typically set equal to $1$.

However, while this sum over all states is conserved, the same need not be true
for the sum of $\pi_i(t)$ over $i$ in a subset $Y \subset X$.   This poses 
a challenge to studying a Markov process as built from smaller parts: the parts are
not themselves Markov processes.  The solution is to describe them as `open' Markov 
processes.   These are a generalization in which probability 
can enter or leave from certain states designated as inputs and outputs:
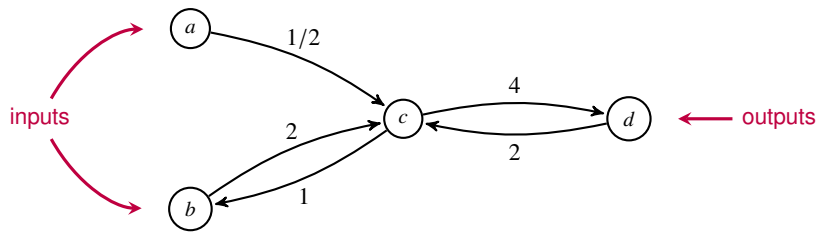
\[
\begin{tikzpicture}[->,>=stealth',shorten >=1pt,thick,scale=1]
  \node[main node] (1) at (0,2.2) {$\scriptstyle{a}$};
  \node[main node](2) at (0,-.2) {$\scriptstyle{b}$};
  \node[main node](3) at (2.83,1)  {$\scriptstyle{c}$};
  \node[main node](4) at (5.83,1) {$\scriptstyle{d}$};
\node(input) [color=purple] at (-2,1) {\small{\textsf{inputs}}};
\node(output) [color=purple] at (7.83,1) {\small{\textsf{outputs}}};
  \path[every node/.style={font=\sffamily\small}, shorten >=1pt]
    (3) edge [bend left=12] node[above] {$4$} (4)
    (4) edge [bend left=12] node[below] {$2$} (3)
    (2) edge [bend left=12] node[above] {$2$} (3) 
    (3) edge [bend left=12] node[below] {$1$} (2)
    (1) edge [bend left=12] node[above] {$1/2$}(3);
    
\path[color=purple, very thick, shorten >=10pt, ->, >=stealth] (output) edge (4);
\path[color=purple, very thick, shorten >=10pt, ->, >=stealth, bend left] (input) edge (1);
\path[color=purple, very thick, shorten >=10pt, ->, >=stealth, bend right]
(input) edge (2);
\end{tikzpicture}
\]
In an open Markov process, probabilities change with time according to the `open master equation', a generalization of the master equation that includes inflows and outflows.  
In the above example, the open master equation is 
\[    \frac{d}{dt}\left[\begin{array}{r} \pi_a(t) \\ \pi_b(t) \\ \pi_c(t) \\ \pi_d(t) \end{array}\right]  \;\; = \; \; \left[\begin{array}{rrrr}
    -1/2    & 0    & 0    & 0  \\
      0                   & -2   & 1   & 0 \\
      1/2   & 2     & -5  & 2 \\
      0                  & 0      & 4   & -2 \\
\end{array}\right] 
\left[\begin{array}{r} \pi_a(t) \\ \pi_b(t) \\ \pi_c(t) \\ \pi_d(t) \end{array}\right]  \; + \;
\left[\begin{array}{c} I_a(t) \\ I_b(t) \\0 \\ 0 \end{array}\right] \; - \; 
\left[\begin{array}{c} 0 \\ 0 \\0 \\ O_d(t) \end{array}\right] .
\]
To the master equation we have added a term describing inflows at the states $a$ and $b$ and subtracted a term describing outflows at the state $d$.   The functions $I_a, I_b$ and $O_d$ are
not part of the data of the open Markov process.  Rather, they are arbitrary smooth real-valued functions of time.  We think of these as provided from outside---for example, though not necessarily, from the rest of a larger Markov process of which the given open Markov process is part.

Open Markov processes can be seen as morphisms in a category, since we can compose two open Markov processes by identifying the outputs of the first with the inputs of the second.   Composition lets us build a Markov process from smaller open parts---or conversely, analyze the behavior of a Markov process in terms of its parts.   The resulting category has been studied in a number of papers \cite{BF,BFP,FongThesis,PollardThesis}, but here we go further and introduce a double category to describe coarse-graining.

`Coarse-graining' is a widely used method of simplifying a Markov process by mapping its set of states $X$ onto some smaller set $X'$ in a manner that respects, or at least approximately respects, the dynamics \cite{And,Buchholz}.  Here we introduce coarse-graining for open Markov processes.   We show how to extend this notion to the case of maps $p \maps X \to X'$ that are not surjective, obtaining a general concept of morphism between open Markov processes.

Since open Markov processes are already morphisms in a category, it is natural to treat morphisms between them as morphisms between morphisms, or `2-morphisms'.  We can do this using double categories.   These were first introduced by Ehresmann \cite{Ehresmann63, Ehresmann65}, and they have long been used in topology and other branches of pure mathematics \cite{Brown1,Brown2}.  More recently they have been used to study open dynamical systems \cite{LS} and open discrete-time Markov chains \cite{Panan}.  So, it should not be surprising that they are also useful for open Markov processes.

A 2-morphism in a double category looks like this:
\[
\begin{tikzpicture}[scale=1]
\node (D) at (-4,0.5) {$A$};
\node (E) at (-2,0.5) {$B$};
\node (F) at (-4,-1) {$C$};
\node (A) at (-2,-1) {$D$};
\node (B) at (-3,-0.25) {$\Downarrow \alpha$};
\path[->,font=\scriptsize,>=angle 90]
(D) edge node [above]{$M$}(E)
(E) edge node [right]{$g$}(A)
(D) edge node [left]{$f$}(F)
(F) edge node [above]{$N$} (A);
\end{tikzpicture}
\]
While a mere category has only objects and morphisms, here we have a few more types of entities. We call $A, B, C$ and $D$ `objects', $f$ and $g$ `vertical 1-morphisms', $M$ and $N$ `horizontal 1-cells', and $\alpha$ a `2-morphism'.   We can compose vertical 1-morphisms to get new vertical 1-morphisms and compose horizontal 1-cells to get new horizontal 1-cells.  We can compose the 2-morphisms in two ways: horizontally by setting squares side by side, and vertically by setting one on top of the other.   In a `strict' double category all these forms of composition are associative.  In a `pseudo' double category, horizontal 1-cells compose in a weakly associative manner: that is, the associative law holds only up to an invertible 2-morphism, called the `associator', which obeys a coherence law.  This is just a quick sketch of the ideas; for full definitions see for example the works of Grandis and Par\'e \cite{GP1,GP2}.

 We construct a double category $\MMark$ with:
\begin{enumerate}
\item finite sets as objects, 
\item maps between finite sets as vertical 1-morphisms,
\item open Markov processes as horizontal 1-cells, 
\item morphisms between open Markov processes as 2-morphisms.
\end{enumerate}
Composition of open Markov processes is only weakly associative, so this is a
pseudo double category.  

The plan of the paper is as follows.  In Section \ref{sec:markov_processes} we define open Markov processes and steady state solutions of the open master equation.  In Section \ref{sec:coarse-graining} we introduce coarse-graining first for Markov processes and then open Markov processes. In Section \ref{sec:MMark} we construct the double category $\MMark$ described above.  We prove this is a symmetric monoidal double category in the sense of Shulman \cite{Shulman}.  This captures the fact that we can not only compose open Markov processes but also `tensor' them by setting them side by side.  For example, if we compose this open Markov process:
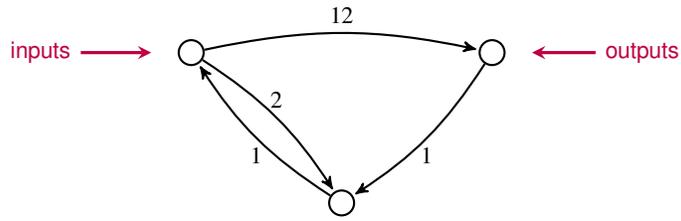
\[
\begin{tikzpicture}[->,>=stealth',shorten >=1pt,thick,scale=1]
  \node[main node] (4) at (0,1) {$$};
  \node[main node](5) at (2,-1) {$$};
  \node[main node](6) at (4,1) {$$};
\node(input) [color=purple] at (-2,1) {\small{\textsf{inputs}}};
\node(output) [color=purple] at (6,1) {\small{\textsf{outputs}}};
  \path[every node/.style={font=\sffamily\small}, shorten >=1pt]
    (4) edge [bend left=12] node[above] {$2$} (5)
    (4) edge [bend left=12] node[above] {$12$} (6)
    (5) edge [bend left=12] node[below] {$1$} (4) 
    (6) edge [bend left=12] node[below] {$1$} (5);
    
\path[color=purple, very thick, shorten >=10pt, ->, >=stealth] (output) edge (6);
\path[color=purple, very thick, shorten >=10pt, ->, >=stealth] (input) edge (4);
\end{tikzpicture}
\]
with the one shown before:
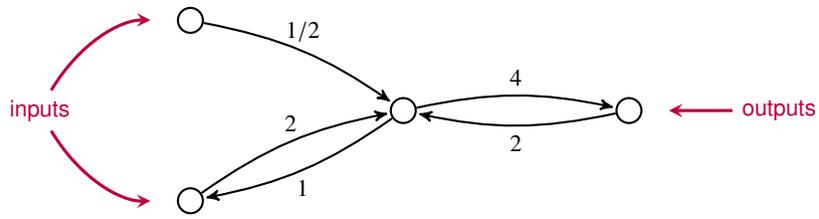
\[
\begin{tikzpicture}[->,>=stealth',shorten >=1pt,thick,scale=1]
  \node[main node] (1) at (0,2.2) {$$};
  \node[main node](2) at (0,-.2) {$$};
  \node[main node](3) at (2.83,1)  {$$};
  \node[main node](4) at (5.83,1) {$$};
\node(input) [color=purple] at (-2,1) {\small{\textsf{inputs}}};
\node(output) [color=purple] at (7.83,1) {\small{\textsf{outputs}}};
  \path[every node/.style={font=\sffamily\small}, shorten >=1pt]
    (3) edge [bend left=12] node[above] {$4$} (4)
    (4) edge [bend left=12] node[below] {$2$} (3)
    (2) edge [bend left=12] node[above] {$2$} (3) 
    (3) edge [bend left=12] node[below] {$1$} (2)
    (1) edge [bend left=12] node[above] {$1/2$}(3);
    
\path[color=purple, very thick, shorten >=10pt, ->, >=stealth] (output) edge (4);
\path[color=purple, very thick, shorten >=10pt, ->, >=stealth, bend left] (input) edge (1);
\path[color=purple, very thick, shorten >=10pt, ->, >=stealth, bend right]
(input) edge (2);
\end{tikzpicture}
\]
we obtain this open Markov process:
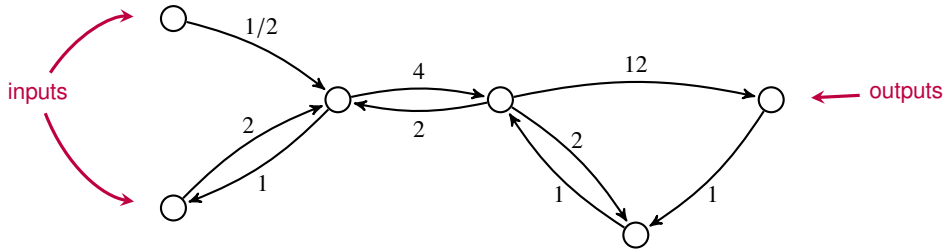
\[
\begin{tikzpicture}[->,>=stealth',shorten >=1pt,thick,scale=.9]
  \node[main node] (1) at (0,2.2) {$$};
  \node[main node](2) at (0,-.6) {$$};
  \node[main node](3) at (2.43,1)  {$$};
  \node[main node](4) at (4.83,1) {$$};
  \node[main node](5) at (6.83,-1) {$$};
  \node[main node](6) at (8.83,1) {$$};
\node(input) [color=purple] at (-2,1.1) {\small{\textsf{inputs}}};
\node(output) [color=purple] at (10.83,1.1) {\small{\textsf{outputs}}};
  \path[every node/.style={font=\sffamily\small}, shorten >=1pt]
    (3) edge [bend left=12] node[above] {$4$} (4)
    (4) edge [bend left=12] node[below] {$2$} (3)
    (2) edge [bend left=12] node[above] {$2$} (3) 
    (3) edge [bend left=12] node[below] {$1$} (2)
    (1) edge [bend left=12] node[above] {$1/2$}(3) 
    (4) edge [bend left=12] node[above] {$2$} (5)
    (4) edge [bend left=12] node[above] {$12$} (6)
    (5) edge [bend left=12] node[below] {$1$} (4) 
    (6) edge [bend left=12] node[below] {$1$} (5);
    
\path[color=purple, very thick, shorten >=10pt, ->, >=stealth] (output) edge (6);
\path[color=purple, very thick, shorten >=10pt, ->, >=stealth, bend left] (input) edge (1);
\path[color=purple, very thick, shorten >=10pt, ->, >=stealth, bend right]
(input) edge (2);
\end{tikzpicture}
\]
but if we tensor them, we obtain this:
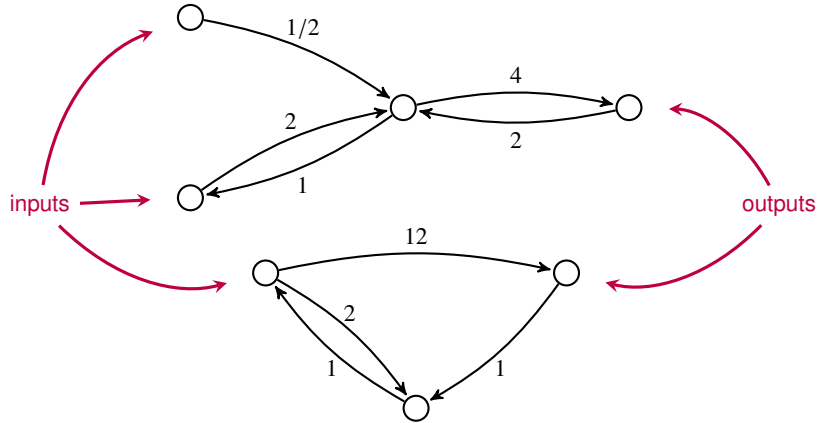
\[
\begin{tikzpicture}[->,>=stealth',shorten >=1pt,thick,scale=1]
  \node[main node] (1) at (0,2.2) {$$};
  \node[main node](2) at (0,-.2) {$$};
  \node[main node](3) at (2.83,1)  {$$};
  \node[main node](4) at (5.83,1) {$$};
  \node[main node](4a) at (1,-1.2) {$$};
  \node[main node](5) at (3,-3) {$$};
  \node[main node](6) at (5,-1.2) {$$};
\node(input)[color=purple] at (-2,-.3) {\small{\textsf{inputs}}};
\node(output) [color=purple] at (7.83,-.3) {\small{\textsf{outputs}}};
  \path[every node/.style={font=\sffamily\small}, shorten >=1pt]
    (3) edge [bend left=12] node[above] {$4$} (4)
    (4) edge [bend left=12] node[below] {$2$} (3)
    (2) edge [bend left=12] node[above] {$2$} (3) 
    (3) edge [bend left=12] node[below] {$1$} (2)
    (1) edge [bend left=12] node[above] {$1/2$}(3) 
    (4a) edge [bend left=12] node[above] {$2$} (5)
    (4a) edge [bend left=12] node[above] {$12$} (6)
    (5) edge [bend left=12] node[below] {$1$} (4a) 
    (6) edge [bend left=12] node[below] {$1$} (5);
    
\path[color=purple, very thick, shorten >=10pt, ->, >=stealth, bend right] (output) edge (4);
\path[color=purple, very thick, shorten >=10pt, ->, >=stealth, bend left] (input) edge (1);
\path[color=purple, very thick, shorten >=10pt, ->, >=stealth] (input) edge (2);
\path[color=purple, very thick, shorten >=10pt, ->, >=stealth, bend left] (output) edge (6);
\path[color=purple, very thick, shorten >=10pt, ->, >=stealth, bend right] (input) edge (4a);
\end{tikzpicture}
\]

If we fix constant probabilities at the inputs and outputs, there typically exist solutions of the open master equation with these boundary conditions that are constant as a function of time.  These are called `steady states'.  Often these are \emph{nonequilibrium} steady states, meaning that there is a nonzero net flow of probabilities at the inputs and outputs.   For example, probability can flow through an open Markov process at a constant rate in a nonequilibrium steady state.

In previous work, Fong, Pollard and the first author studied the relation between probabilities and flows at the inputs and outputs that holds in steady state \cite{BFP,BP}.  They called the process of extracting this relation from an open Markov process `black-boxing', since it gives a way to forget the internal workings of an open system and remember only its externally observable behavior.   They proved that black-boxing is compatible with composition and tensoring.  This result can be summarized by saying that black-boxing is a symmetric monoidal functor. 

In Section \ref{sec:black-boxing} we show that black-boxing is compatible with morphisms between open Markov processes.  To make this idea precise, we prove that black-boxing gives a map from the double category $\MMark$ to another double category, called $\LLinRel$, which has:
\begin{enumerate}
\item finite-dimensional real vector spaces $U,V,W,\dots$ as objects,
\item linear maps $f \maps V \to W$ as vertical 1-morphisms from $V$ to $W$,
\item linear relations $R \subseteq V \oplus W$  as horizontal 1-cells from $V$ to $W$,
\item squares 
\[
\begin{tikzpicture}[scale=1.5]
\node (D) at (-4,0.5) {$V_1$};
\node (E) at (-2,0.5) {$V_2$};
\node (F) at (-4,-1) {$W_1$};
\node (A) at (-2,-1) {$W_2$};
\node (B) at (-3,-0.25) {};
\path[->,font=\scriptsize,>=angle 90]
(D) edge node [above]{$R \subseteq V_1 \oplus V_2$}(E)
(E) edge node [right]{$g$}(A)
(D) edge node [left]{$f$}(F)
(F) edge node [above]{$S \subseteq W_1 \oplus W_2$} (A);
\end{tikzpicture}
\]
obeying $(f \oplus g)R \subseteq S$ as 2-morphisms. 
\end{enumerate}
Here a `linear relation' from a vector space $V$ to a vector space $W$ is a linear subspace
$R \subseteq V \oplus W$.   Linear relations can be composed in the same way as relations \cite{BE}.  The double category $\LLinRel$ becomes symmetric monoidal using direct sum as the tensor product, but unlike $\MMark$ it is strict: that is, composition of linear relations is associative.

Maps between symmetric monoidal double categories are called `symmetric monoidal
double functors' \cite{Courser}.  Our main result, Thm.\ \ref{thm:main}, 
says that black-boxing gives a symmetric monoidal double functor
\[   \blacksquare \maps \MMark \to \LLinRel .\]
The hardest part is to show that black-boxing preserves composition of horizontal
1-cells: that is, black-boxing a composite of open Markov processes gives the composite 
of their black-boxings.  Luckily, for this we can adapt a previous argument \cite{BP}.  
Thus, the new content of this result concerns the vertical 1-morphisms and especially 
the 2-morphisms, which describe coarse-grainings.

An alternative approach to studying morphisms between open Markov processes uses bicategories rather than double categories \cite{Benabou,Stay}.  In Section \ref{sec:bicat} we use a result of Shulman \cite{Shulman} to construct symmetric monoidal bicategories $\bold{Mark}$ and $\bold{LinRel}$ from the symmetric monoidal double categories $\MMark$ and $\LLinRel$.  We conjecture that the black-boxing double functor determines a functor between these symmetric monoidal bicategories.  However, double categories seem to be a simpler framework for coarse-graining open Markov processes.

It is worth comparing some related work.   Fong, Pollard and the first author constructed a symmetric monoidal category where the morphisms are open Markov processes \cite{BFP,BP}.  Like us, they only consider Markov processes where time is continuous and the set of states is finite.  
However, they formalized such Markov processes in a slightly different way than we do here: they defined a Markov process to be a directed multigraph where each edge is assigned a positive number called its `rate constant'.   In other words, they defined it to be a diagram
\[ \xymatrix{ (0,\infty) & E \ar[l]_-r \ar[r]<-.5ex>_t  \ar[r] <.5ex>^s & X  }  \]
where $X$ is a finite set of vertices or `states', $E$ is a finite set of edges or `transitions' 
between states, the functions $s,t \maps E \to X$ give the source and target of each edge, and 
$r \maps E \to (0,\infty)$ gives the rate constant of each edge.   They explained how from this data
one can extract a matrix of real numbers $(H_{ij})_{i,j \in X}$ called the `Hamiltonian' of the Markov process, with two familiar properties:
\begin{enumerate}
\item $H_{ij} \geq 0$ if $i \neq j$, 
\item $\sum_{i \in X} H_{ij} = 0$ for all $j \in X$.
\end{enumerate}
A matrix with these properties is called `infinitesimal stochastic', since these conditions are equivalent to $\exp(tH)$ being stochastic for all $t \ge 0$.    

In the present work we skip the directed multigraphs and work directly with the Hamiltonians.  Thus, we define a Markov process to be a finite set $X$ together with an infinitesimal stochastic matrix $(H_{ij})_{i,j \in X}$.  This allows us to work more directly with the Hamiltonian and the all-important master equation
\[         \frac{d}{dt} \pi(t) = H \pi(t)  \]
which describes the evolution of a time-dependent probability distribution $\pi(t) \maps X \to \R$.

Clerc, Humphrey and Panangaden have constructed a bicategory \cite{Panan} with finite sets as objects, `open discrete labeled Markov processes' as morphisms, and `simulations' as 2-morphisms. In their framework, `open' has a similar meaning as it does in works listed above.  These open discrete labeled Markov processes are also equipped with a set of `actions' which represent interactions between the Markov process and the environment, such as an outside entity acting on a stochastic system.   A `simulation' is then a function between the state spaces that map the inputs, outputs and set of actions of one open discrete labeled Markov process to the inputs, outputs and set of actions of another.   

Another compositional framework for Markov processes is given by de Francesco Albasini, Sabadini and Walters \cite{Francesco} in which they construct an algebra of `Markov automata'.  A Markov automaton is a family of matrices with nonnegative real coefficients that is indexed by elements of a binary product of sets, where one set represents a set of `signals on the left interface' of the Markov automata and the other set analogously for the right interface.

\subsection*{Notation and Terminology}

Following Shulman, we use `double category' to mean `pseudo double category', and use `strict double category' to mean a double category for which horizontal composition is strictly associative and unital.   (In older literature, `double category' often refers to a strict double category.)  

It is common to use blackboard bold for the first letter of the name of a double category, and we do so here.  Bicategories are written in boldface, while ordinary categories are written in sans serif font.
Thus, three main players in this paper are a double category $\MMark$, a bicategory $\bold{Mark}$, and a category $\Mark$, all closely related.

\section{Open Markov processes}
\label{sec:markov_processes}

Before explaining open Markov processes we should recall a bit about Markov processes.  As mentioned in the Introduction, we use `Markov process' as a short term for `continuous-time Markov process with a finite set of states', and we identify any such Markov process with the
infinitesimal stochastic matrix appearing in its master equation.  We make this precise with a bit 
of terminology that is useful throughout the paper. 

Given a finite set $X$, we call a function $v \maps X \to \R$ a `vector' and call its values at points $x \in X$ its `components' $v_x$.   We define a `probability distribution' on $X$ to be a vector $\pi \maps X \to \R$ whose components are nonnegative and sum to $1$.  As usual, we use $\R^X$ to denote the vector space of functions $v \maps X \to \R$.    Given a linear operator $T \maps \R^X \to \R^Y$ we have $(T v)_i = \sum_{j \in X} T_{ij} v_j$ for some `matrix' $T \maps Y \times X \to \R$ with entries $T_{ij}$.   

\begin{defn}
Given a finite set $X$, a linear operator $H \maps \R^X \to \R^X$ is \define{infinitesimal stochastic} if
\begin{enumerate}
\item $H_{ij} \geq 0$ for $i \neq j$ and 
\item $\sum_{i \in X} H_{ij}=0$ for each $j \in X$.
\end{enumerate}
\end{defn}

The reason for being interested in such operators is that when exponentiated
they give stochastic operators.

\begin{defn} Given finite sets $X$ and $Y$, a linear operator $T \maps \R^X \to \R^Y$ is \define{stochastic} if for any probability distribution $\pi$ on $X$, $T\pi$ is a probability
distribution on $Y$.
\end{defn}

Equivalently, $T$ is stochastic if and only if 
\begin{enumerate}
\item $T_{ij} \ge 0$ for all $i \in Y$, $j \in X$ and
\item $\sum_{i \in Y} T_{ij} = 1$ for each $j \in X$.
\end{enumerate}
If we think of $T_{ij}$ as the probability for $j \in X$ to be mapped to $i \in Y$, these conditions
make intuitive sense.   Since stochastic operators are those that preserve probability distributions, the composite of stochastic operators is stochastic.  

In Lemma \ref{lem:exponentiation} we recall that a linear operator $H \maps \R^X \to \R^X$ 
is infinitesimal stochastic if and only if its exponential
\[     \exp(tH) = \sum_{n = 0}^\infty \frac{(tH)^n}{n!}  \]
is stochastic for all $t \ge 0$.  
Thus, given an infinitesimal stochastic operator $H$, for any time $t \ge 0$
we can apply the operator $\exp(tH) \maps \R^X \to \R^X$ to any probability 
distribution $\pi \in \R^X$ and get a probability distribution
\[      \pi(t) = \exp(tH) \pi. \]
These probability distributions $\pi(t)$ obey the \define{master equation}
\[       \frac{d}{dt}\pi(t) = H \pi(t) .\]
Moreover, any solution of the master equation arises this way.

All the material so far is standard \cite[Sec.\ 2.1]{Norris}.  We now turn to open Markov processes.  

\begin{defn}
We define a \define{Markov process} to be a pair $(X,H)$ where $X$ is a finite set and 
$H \maps \R^X \to \R^X$ is an infinitesimal stochastic operator.    We also call $H$ a 
Markov process \define{on} $X$. 
\end{defn}

\begin{defn}
\label{defn:open_markov_process}
We define an \define{open Markov process} to consist of finite sets $X$, $S$ and $T$ and
injections
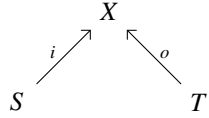
\[
\begin{tikzpicture}[scale=1.2]
\node (D) at (-4,-0.5) {$S$};
\node (E) at (-3,.5) {$X$};
\node (F) at (-2,-0.5) {$T$};
\path[->,font=\scriptsize,>=angle 90]
(D) edge node[above,pos=0.3] {$i$}(E)
(F) edge node[above,pos=0.3] {$o$}(E);
\end{tikzpicture}
\]
together with a Markov process $(X,H)$.   We call $S$ the set of \define{inputs} and
$T$ the set of \define{outputs}.
\end{defn}

In general, a diagram of this shape in any category:
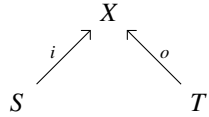
\[
\begin{tikzpicture}[scale=1.2]
\node (D) at (-4,-0.5) {$S$};
\node (E) at (-3,.5) {$X$};
\node (F) at (-2,-0.5) {$T$};
\path[->,font=\scriptsize,>=angle 90]
(D) edge node[above,pos=0.3] {$i$}(E)
(F) edge node[above,pos=0.3] {$o$}(E);
\end{tikzpicture}
\]
is called a \define{cospan}.  The objects $S$ and $T$ are called the \define{feet}, the
object $X$ is called the \define{apex}, and the morphisms $i$ and $o$ are called the \define{legs}.  We use $\FinSet$ to stand for the category of finite sets and functions.  Thus, an open Markov process is a cospan in $\FinSet$ with injections as legs and a Markov process on its apex.  We do not require that the injections have disjoint range.  We often abbreviate an open Markov process as
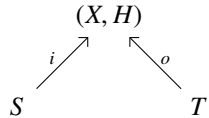
\[
\begin{tikzpicture}[scale=1.2]
\node (D) at (-4,-0.5) {$S$};
\node (E) at (-3,.5) {$(X,H)$};
\node (F) at (-2,-0.5) {$T$};
\path[->,font=\scriptsize,>=angle 90]
(D) edge node[above,pos=0.3] {$i$}(E)
(F) edge node[above,pos=0.3] {$o$}(E);
\end{tikzpicture}
\]
or simply $S \stackrel{i}{\rightarrow} (X,H) \stackrel{o}{\leftarrow} T$.

Given an open Markov process we can write down an `open' version of the master equation, where probability can also flow in or out of the inputs and outputs.   To work with the open
master equation we need two well-known concepts:

\begin{defn}
\label{defn:push_and_pull}
Let $f \maps A \to B$ be a map between finite sets.  The linear map $f^* \maps \R^B \to
\R^A$ sends any vector $v \in \R^B$ to its \define{pullback} along $f$, given by
\[                       f^*(v) = v \circ f  . \]
The linear map $f_* \maps \R^A \to \R^B$ sends any vector $v \in \R^A$ to its \define{pushforward} along $f$, given by
\[                       (f_*(v))(b) = \sum_{ \{a : \; f(a) = b\} } v(a)  .\]
\end{defn}
\noindent If we write $f^*$ and $f_*$ as matrices with respect to the standard bases of
$\R^A$ and $\R^B$, they are simply transposes of one another.

Now, suppose we are given an open Markov process
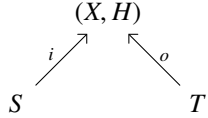
\[
\begin{tikzpicture}[scale=1.2]
\node (D) at (-4,-0.5) {$S$};
\node (E) at (-3,.5) {$(X,H)$};
\node (F) at (-2,-0.5) {$T$};
\path[->,font=\scriptsize,>=angle 90]
(D) edge node[above,pos=0.3] {$i$}(E)
(F) edge node[above,pos=0.3] {$o$}(E);
\end{tikzpicture}
\]
together with \define{inflows} $I \maps \R \to \R^S$ and \define{outflows} $O \maps \R \to \R^T$, arbitrary smooth functions of time.  We write the value of the inflow at $s \in S$ at time $t$ as $I_s(t)$, and similarly for outflows and other functions of time.  We say a function $v \maps \R \to \R^X$ obeys the \define{open master equation} if
\[ \frac{dv(t)}{dt} = H v(t) + i_*(I(t)) - o_*(O(t)). \]
This says that for any state $j \in X$ the time derivative of $v_j(t)$ 
takes into account not only the usual term from the master equation, but also inflows and outflows.

If the inflows and outflows are constant in time, a solution $v$ of the open master equation
that is also constant in time is called a \define{steady state}. More formally:

\begin{defn}
Given an open Markov process $S \stackrel{i}{\rightarrow} (X,H) \stackrel{o}{\leftarrow} T$ together with $I \in \R^S$ and $O \in \R^T$, a \define{steady state} with inflows $I$ and outflows $O$ is an element $v \in \R^X$ such that 
\[     H v + i_*(I) - o_*(O) = 0 .\]   
Given $v \in \R^X$ we call $i^*(v) \in \R^S$ and $o^*(v) \in \R^T$ the 
\define{input probabilities} and \define{output probabilities}, respectively.
\end{defn}

\begin{defn}
\label{defn:black-boxing}
Given an open Markov process $S \stackrel{i}{\rightarrow} (X,H) \stackrel{o}{\leftarrow} T$, we define its \define{black-boxing} to be the set
\[   \blacksquare\big(S \stackrel{i}{\rightarrow} (X,H) \stackrel{o}{\leftarrow} T\big) 
\subseteq \R^S \oplus \R^S \oplus \R^T \oplus \R^T \]
consisting of all 4-tuples $(i^*(v),I,o^*(v),O)$ where $v \in \R^X$ is some steady state
with inflows $I \in \R^S$ and outflows $O \in \R^T$.
\end{defn}

Thus, black-boxing records the relation between input probabilities, inflows, output probabilities and outflows that holds in steady state.  This is the `externally observable steady state behavior' of the open Markov process.   It has already been shown \cite{BFP,BP} that black-boxing can be seen as a functor between categories.    Here we go further and describe it as a double functor between double categories, in order to study the effect of black-boxing on morphisms between open Markov processes.

\section{Morphisms of open Markov processes}
\label{sec:coarse-graining}

There are various ways to approximate a Markov process by another Markov process on a smaller set, all of which can be considered forms of coarse-graining \cite{Buchholz}.  A common approach is to take a Markov process $H$ on a finite set $X$ and a surjection $p \maps X \to X'$ and create a Markov process on $X'$.   In general this requires a choice of `stochastic section' for $p$, defined as follows:

\begin{defn}
Given a function $p \maps X \to X'$ between finite sets, a \define{stochastic section} for
$p$ is a stochastic operator $s \maps \R^{X'} \to \R^X$ such that $p_* s = 1_{X'}$. 
\end{defn}

It is easy to check that a stochastic section for $p$ exists if and only if $p$ is a surjection.  In Lemma \ref{lem:new_markov_process} we show that given a Markov process $H$ on $X$ and a surjection $p \maps X \to X'$, any stochastic section $s \maps \R^{X'} \to \R^X$ gives a Markov process on $X'$, namely
\[            H' =  p_* H s. \]
Experts call the matrix corresponding to $p_*$ the \define{collector matrix}, and they call $s$ the \define{distributor matrix} \cite{Buchholz}.  The names help clarify what is going on.  The collector matrix, coming from the surjection $p \maps X \to X'$, typically maps many states of $X$ to each state of $X'$.  The distributor matrix, the stochastic section $s \maps \R^{X'} \to \R^X$, typically maps each state in $X'$ to a linear combination of many states in $X$.  Thus, $H' = p_* H s$ distributes each state of $X'$, applies $H$, and then collects the results.

In general $H'$ depends on the choice of $s$, but sometimes it does not:

\begin{defn} 
We say a Markov process $H$ on $X$ is \define{lumpable} with respect to 
a surjection $p \maps X \to X'$ if the operator $p_* H s$ is independent of the choice of stochastic section $s \maps \R^{X'} \to \R^X$.
\end{defn}

This concept is not new \cite{Buchholz}.   In Thm.\ \ref{thm:lumpability} we show that it is equivalent to another traditional formulation, and also to an even simpler one: $H$ is lumpable with respect to $p$ if and only if $p_* H = H' p_*$.    This equation has the advantage of making sense even when $p$ is not a surjection.   Thus, we can use it to define a more general concept of morphism between Markov processes:

\begin{defn}  
Given Markov processes $(X,H)$ and $(X',H')$, a \define{morphism of Markov
processes} $p \maps (X,H) \to (X',H')$ is a map $p \maps X \to X'$ such that
$p_* H = H' p_*$.  
\end{defn}

There is a category $\Mark$ with Markov processes as objects and the morphisms as defined above, where composition is the usual composition of functions.  But what is the meaning of such a morphism?   Using Lemma \ref{lem:exponentiation} one can check that for any Markov processes $(X,H)$ and $(X',H')$, and any map
$p \maps X \to X'$, we have
\[    p_* H = H' p_* \; \iff \;  p_* \exp(tH) = \exp(tH') p_* \textrm{ for all } t \ge 0.\]
Thus, $p$ is a morphism of Markov processes if evolving a probability distribution
on $X$ via $\exp(tH)$ and then pushing it forward along $p$ is the same as
pushing it forward and then evolving it via $\exp(tH')$.  

We can also define morphisms between \emph{open} Markov processes:

\begin{defn}
\label{defn:coarse-graining}  A \define{morphism of open Markov processes} from the open Markov process $S \stackrel{i}{\rightarrow} (X,H) \stackrel{o}{\leftarrow} T$ to the open Markov process $S' \stackrel{i'}{\rightarrow} (X',H') \stackrel{o'}{\leftarrow} T'$ is a triple of functions $f \maps S \to S'$, $p \maps X \to X'$, $g \maps T \to T'$ such that the squares in this diagram are pullbacks:
\[
\begin{tikzpicture}[scale=1.5]
\node (D) at (-4.2,-0.5) {$S$};
\node (A) at (-4.2,-2) {$S'$};
\node (B) at (-1.8,-2) {$T'$};
\node (E) at (-3,-0.5) {$X$};
\node (F) at (-1.8,-0.5) {$T$};
\node (G) at (-3,-2) {$X'$};
\path[->,font=\scriptsize,>=angle 90]
(D) edge node[above] {$i$}(E)
(A) edge node[above] {$i'$} (G)
(B) edge node [above]{$o'$} (G)
(F) edge node [above]{$o$}(E)
(D) edge node [left]{$f$}(A)
(F) edge node [right]{$g$}(B)
(E) edge node[left] {$p$}(G);
\end{tikzpicture}
\]
and $p_* H = H' p_*$.
\end{defn}

\noindent
We need the squares to be pullbacks so that in Lemma \ref{lem:black-boxing_2-morphisms} we can black-box morphisms of open Markov processes.  In Lemma \ref{lem:horizontal_composition} we show that horizontally composing these morphisms preserves this pullback property.  But to do this, we need the horizontal arrows in these squares to be injections.  This explains the conditions in Defs.\ \ref{defn:open_markov_process} and \ref{defn:coarse-graining}.

We often abbreviate a morphism of open Markov processes as
\[
\begin{tikzpicture}[scale=1.5]
\node (D) at (-4.2,-0.5) {$S$};
\node (A) at (-4.2,-2) {$S'$};
\node (B) at (-1.8,-2) {$T'$};
\node (E) at (-3,-0.5) {$(X,H)$};
\node (F) at (-1.8,-0.5) {$T$};
\node (G) at (-3,-2) {$(X',H')$};
\path[->,font=\scriptsize,>=angle 90]
(D) edge node[above] {$i_1$}(E)
(A) edge node[above] {$i'_1$} (G)
(B) edge node [above]{$o'_1$} (G)
(F) edge node [above]{$o_1$}(E)
(D) edge node [left]{$f$}(A)
(F) edge node [right]{$g$}(B)
(E) edge node[left] {$p$}(G);
\end{tikzpicture}
\]

As an example, consider the following diagram:
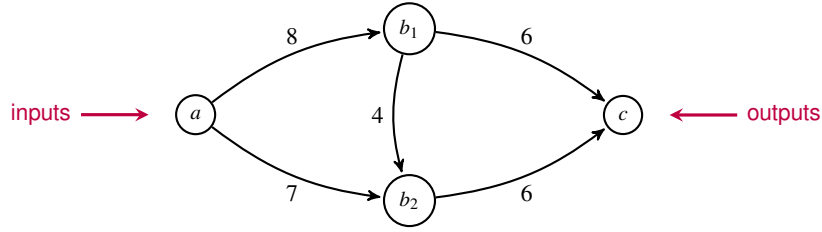
\[
\begin{tikzpicture}[->,>=stealth',shorten >=1pt,auto,node distance=3.7cm,
thick,main node/.style={circle,fill=white!20,draw,font=\sffamily\Large\bfseries},terminal/.style={circle,fill=blue!20,draw,font=\sffamily\Large\bfseries}]]
\node[main node](1) {$\scriptstyle{a}$};
\node[main node](3) [below right=0.7cm and 2.4cm of 1] {$\scriptstyle{b_2}$};
\node[main node](5) [above right=0.7cm and 2.4cm of 1] {$\scriptstyle{b_1}$};
\node[main node](4) [above right=0.7cm and 2.4cm of 3] {$\scriptstyle{c}$};
\node(input) [left=1.25 cm of 1,color=purple] {\small{\textsf{inputs}}};
\node(output) [right=1.25 cm of 4,color=purple] {\small{\textsf{outputs}}};
\path[every node/.style={font=\sffamily\small}, shorten >=1pt]
(3) edge [bend right =15] node[below] {$6$} (4)
(5) edge [bend left =15] node[above] {$6$} (4)
(1) edge [bend left =15] node[above] {$8$} (5)
(5) edge [bend right =15] node[left] {$4$} (3)
(1) edge [bend right =15] node[below] {$7$} (3);
\path[color=purple, very thick, shorten >=10pt, ->, >=stealth] (output) edge (4);
\path[color=purple, very thick, shorten >=10pt, ->, >=stealth] (input) edge (1);
\end{tikzpicture}
\]
This is a way of drawing an open Markov process $S \stackrel{i}{\rightarrow} (X,H) \stackrel{o}{\leftarrow} T$ where $X = \{a,b_1,b_2,c \}$, $S$ and $T$ are one-element sets, $i$ maps the one element of $S$ to $a$, and $o$ maps the one element of $T$ to $c$.   As explained in Section \ref{sec:intro}, we can read off the infinitesimal stochastic operator $H \maps \R^X \to \R^X$ from this diagram and obtain
\[ H=
\left[\begin{array}{rrrr}
  \!  -15    & 0    & 0    & 0  \\
      8     & -10 & 0    & 0 \\
      7     & 4    & -6   & 0 \\
      0     & 6    & 6    & 0 \\
\end{array}\right].
\]
The resulting open master equation is
\[    \frac{d}{dt}\left[\begin{array}{r} v_{a_1}(t) \\ v_{b_1}(t) \\ v_{b_2}(t) \\ v_{c_1}(t) \end{array}\right]  \;\; = \; \; 
\left[\begin{array}{rrrr}
  \!  -15    & 0    & 0    & 0  \\
      8     & -10 & 0    & 0 \\
      7     & 4    & -6   & 0 \\
      0     & 6    & 6    & 0 \\
\end{array}\right]
\left[\begin{array}{r} v_{a_1}(t) \\ v_{b_1}(t) \\ v_{b_2}(t) \\ v_{c_1}(t) \end{array}\right]  \; + \;
\left[\begin{array}{c} I(t) \\ 0 \\0 \\ 0 \end{array}\right] \; - \; 
\left[\begin{array}{c} 0 \\ 0 \\0 \\ O(t) \end{array}\right] .  \]
Here $I$ is an arbitrary smooth 
function of time describing the inflow at the one point of $S$, and $O$ is a similar function
describing the outflow at the one point of $T$.

Suppose we want to simplify this open Markov process by identifying the states $b_1$ and
$b_2$.    To do this we take $X' = \{a,b,c \}$ and define $p \maps X \to X'$ by
\[    p(a) = a, \quad p(b_1) = p(b_2) = b, \quad p(c) = c .\]
To construct the infinitesimal stochastic operator $H' \maps \R^{X'} \to \R^{X'}$ for the
simplified open Markov process we need to choose a stochastic section $s \maps \R^{X'} \to \R^X$ for $p$, for example
\[ s=
\left[\begin{array}{rrr}
1 & 0 & 0  \\
0 & 1/3 & 0 \\
0 & 2/3 & 0 \\
0 & 0 & 1  \\
\end{array}\right].
\]
This says that if our simplified Markov process is in the state $b$, we assume the original Markov
process has a $1/3$ chance of being in state $b_1$ and a $2/3$ chance of being in state $b_2$. 
The operator $H' = p_* H s$ is then
\[ H' =
\left[\begin{array}{rrr}
    -15    & 0   & 0 \\
    15     & -6& 0\\
    0       & 6 & 0\\
    \end{array}\right] .
\]
It may be difficult to justify the assumptions behind our choice of stochastic section, but 
the example at hand has a nice feature: $H'$ is actually independent of this choice.   In other words, $H$ is lumpable with respect to $p$.  The reason is explained in Thm.\ \ref{thm:lumpability}.  Suppose we partition $X$ into blocks, each the inverse image of some point of $X'$.   Then $H$ is lumpable with respect to $p$ if and only if when we sum the rows in each block of $H$, all the columns within any given block of the resulting matrix are identical.  This matrix is $p_* H$:
\[ H=
\left[\begin{array}{r;{2pt/2pt}rr;{2pt/2pt}r}
  \!  -15    & 0    & 0    & 0  \\ \hdashline[2pt/2pt]
      8     & -10 & 0    & 0 \\ 
      7     & 4    & -6   & 0 \\  \hdashline[2pt/2pt]
      0     & 6    & 6    & 0 \\ 
\end{array}\right]  \implies 
p_* H = 
\left[\begin{array}{r;{2pt/2pt}rr;{2pt/2pt}r}
  \!  -15    & 0    & 0    & 0  \\ \hdashline[2pt/2pt]
      15     & -6 & -6    & 0 \\  \hdashline[2pt/2pt]
      0     & 6    & 6    & 0 \\ 
\end{array}\right]. 
\] 
While coarse-graining is of practical importance even in the absence of lumpability, the lumpable case is better behaved, so we focus on this case. 

So far we have described a morphism of Markov processes $p \maps (X,H) \to (X',H')$, but 
together with identity functions on the inputs $S$ and outputs $T$ this defines a morphism
of open Markov processes, going from the above open Markov process to this one:
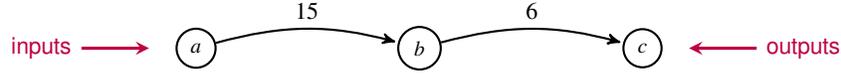
\[
\begin{tikzpicture}[->,>=stealth',shorten >=1pt,auto,node distance=3.7cm,
  thick,main node/.style={circle,fill=white!20,draw,font=\sffamily\Large\bfseries},terminal/.style={circle,fill=blue!20,draw,font=\sffamily\Large\bfseries}]]
  \node[main node](1) {$\scriptstyle{a}$};
  \node[main node](3) [right=2.4cm of 1]  {$\scriptstyle{b}$};
  \node[main node](4) [right=2.4cm of 3] {$\scriptstyle{c}$};
\node(input) [left=1.25 cm of 1,color=purple] {\small{\textsf{inputs}}};
\node(output) [right=1.25 cm of 4,color=purple] {\small{\textsf{outputs}}};
\path[every node/.style={font=\sffamily\small}, shorten >=1pt]
    (3) edge [bend left =15] node[above] {$6$} (4) 
    (1) edge [bend left =15] node[above] {$15$} (3);
\path[color=purple, very thick, shorten >=10pt, ->, >=stealth] (output) edge (4);
\path[color=purple, very thick, shorten >=10pt, ->, >=stealth] (input) edge (1);
\end{tikzpicture}
\]
The open master equation for this new coarse-grained open Markov process is
\[    \frac{d}{dt}\left[\begin{array}{r} v_a(t) \\ v_b(t) \\ v_c(t) \end{array}\right]  \;\; = \; \; 
\left[\begin{array}{rrr}
    -15    & 0   & 0 \\
    15     & -6& 0\\
    0       & 6 & 0\\
    \end{array}\right]
\left[\begin{array}{r} v_a(t) \\ v_b(t) \\ v_c(t) \end{array}\right]  \; + \;
\left[\begin{array}{c} I(t) \\ 0 \\ 0 \end{array}\right] \; - \; 
\left[\begin{array}{c} 0 \\ 0 \\ O(t) \end{array}\right] .  \]

In Section \ref{sec:MMark} we construct a double category $\MMark$ with open Markov processes as horizontal 1-cells and morphisms between these as 2-morphisms.  This double category is our main object of study.   First, however, we should prove the results mentioned above.   For this it is helpful to recall a few standard concepts:

\begin{defn}
A \define{1-parameter semigroup of operators} is a collection of linear operators $U(t) \maps V \to V$ on a vector space $V$, one for each $t \in [0,\infty)$, such that
\begin{enumerate}
\item $U(0)=1$ and 
\item $U(s+t)=U(s)U(t)$ for all $s,t \in [0,\infty)$.
If $V$ is finite-dimensional we say the collection $U(t)$ is \define{continuous} if $t \mapsto U(t) v$ is continuous for each $v \in V$.
\end{enumerate}
\end{defn}

\begin{defn}
Let $X$ be a finite set. A \define{Markov semigroup} is a continuous 1-parameter semigroup
$U(t) \maps \R^X \to \R^X$ such that $U(t)$ is stochastic for each $t \in [0,\infty)$.
\end{defn}

\begin{lem}
\label{lem:exponentiation}
Let $X$ be a finite set and $U(t) \maps \R^X \to \R^X$ a Markov semigroup. Then $U(t)= \exp(tH)$ for a unique infinitesimal stochastic operator $H \maps \R^X \to \R^X$, which is given by
\[  Hv = \left. \frac{d}{dt} U(t)v \right|_{t=0} \]
for all $v \in \R^X$. Conversely, given an infinitesimal stochastic operator $H$, then $\exp(tH)=U(t)$ is a Markov semigroup. 
\end{lem}

\begin{proof}  This is well-known.   For a proof that every continuous one-parameter semigroup of 
operators $U(t)$ on a finite-dimensional vector space $V$ is in fact differentiable and of the form $\exp(tH)$ where $Hv = \left. \frac{d}{dt} U(t)v \right|_{t=0}$, see Engel and Nagel \cite[Sec.\ I.2]{EngelNagel}.  For a proof that $U(t)$ is then a Markov semigroup if and only if $H$ is infinitesimal stochastic, see Norris \cite[Thm.\ 2.1.2]{Norris}. \end{proof}

\begin{lem}
\label{lem:differentiation}
Let $U(t) \maps \R^X \to \R^X$ be a differentiable family of stochastic operators defined for $t \in [0,\infty)$ and having $U(0)=1$. Then $\left. \frac{d}{dt} U(t) \right|_{t=0}$ is infinitesimal stochastic.
\end{lem}

\begin{proof}
Let $H=\left. \frac{d}{dt} U(t) \right|_{t=0}  = \lim_{t \to 0^+} (U(t)-1)/t$. As $U(t)$ is stochastic, its entries are nonnegative and the column sum of any particular column is 1. Then the column sum of any particular column of $U(t)-1$ will be 0 with the off-diagonal entries being nonnegative. Thus $U(t)-1$ is infinitesimal stochastic for all $t \geq 0$, as is $(U(t)-1)/t$, from which it follows that $\lim_{t\to 0^+} (U(t)-U(0))/t = H$ is infinitesimal stochastic.
\end{proof}

\begin{lem}
\label{lem:new_markov_process}
Let $p \maps X \to X'$ be a function between finite sets with a stochastic section $s \maps \R^{X'} \to \R^X$, and let $H \maps \R^X \to \R^X$ be an infinitesimal stochastic operator. Then $H' = p_* H s \maps \R^{X'} \to \R^{X'}$ is also infinitesimal stochastic.
\end{lem}

\begin{proof}
Lemma \ref{lem:exponentiation} implies that $\exp(tH)$ is stochastic for all $t \ge 0$.    For any map $p \maps X \to X'$ the operator $p_* \maps \R^X \to \R^{X'}$ is easily seen to be stochastic, and $s$ is stochastic by assumption.  Thus, $U(t) = p_* \exp(tH) s$ is stochastic for all $t \ge 0$.    Differentiating, we conclude that 
\[   \left. \frac{d}{dt} U(t) \right|_{t=0}= \left. \frac{d}{dt} p_* \exp(tH) s \right|_{t=0} =
\left. p_* \exp(tH)H s \right|_{t=0} = p_* H s \] 
is infinitesimal stochastic by Lemma \ref{lem:differentiation}.
\end{proof}

We can now give some conditions equivalent to lumpability.   The third is widely found in the literature \cite{Buchholz} and the easiest to check in examples.   It makes use of  the standard basis vectors $e_j \in \R^X$ associated to the elements $j$ of any finite set $X$.    The surjection $p \maps X \to X'$ defines a partition on $X$ where two states $j,j' \in X$ lie in the same block of the partition if and only if $p(j) = p(j')$.   The elements of $X'$ correspond to these blocks.   The third condition for lumpability says that $p_* H$ has the same effect on two basis vectors $e_j$ and $e_{j'}$ when $j$ and $j'$ are in the same block.  As mentioned in the example above, this condition
says that if we sum the rows in each block of $H$, all the columns in any given block of the resulting matrix $p_* H$ are identical.

\begin{thm}
\label{thm:lumpability}
Let $p \maps X \to X'$ be a surjection of finite sets and let $H$ be a Markov process on $X$.  Then the following conditions are equivalent:
\begin{enumerate}
\item $H$ is lumpable with respect to $p$.
\item There exists a linear operator $H' \maps \R^{X'} \to \R^{X'}$ such that $p_* H = H' p_*$.
\item $p_* H e_j = p_* H e_{j'}$ for all $j,j' \in X$ such that $p(j)=p(j')$.   
\end{enumerate}
When these conditions hold there is a unique operator $H' \maps \R^{X'} \to \R^{X'}$ 
such that $p_* H = H' p_*$, it is given by $H' = p_* H s$ for any stochastic section $s$ of $p$, and it is infinitesimal stochastic.
\end{thm}

\begin{proof}
$(i) \implies (iii)$.    Suppose that $H$ is lumpable with respect to $p$.  Thus, $p_* H s \maps \R^{X'} \to \R^{X'}$ is independent of the choice of stochastic section $s \maps \R^{X'} \to \R^X$.    Such a stochastic section is simply an arbitrary linear operator that maps each basis vector $e_i \in \R^{X'}$ to a probability distribution on $X$ supported on the set $\{j \in X: p(j) = i\}$.   Thus, for any $j,j' \in X$ with $p(j)=p(j')=i$, we can find stochastic sections $s,s' \maps \R^{X'} \to \R^X$ such that $s(e_i) = e_j$ and $s'(e_i)=e_{j'}$.   Since $p_* Hs = p_* Hs'$, we have
\[  p_* H e_j = p_* Hs(e_i)=p_* Hs'(e_i) = p_* H e_{j'}. \]

$(iii) \implies (ii)$.  Define $H' \maps \R^{X'} \to \R^{X'}$ on basis vectors
$e_i \in \R^{X'}$ by setting 
\[      H' e_i = p_* H e_j \]
for any $j$ with $p(j)=i$.   Note that $H'$ is well-defined: since $p$ is a surjection such $j$ exists, and since $H$ is lumpable, $H'$ is independent of the choice of such $j$.    Next, note that for any $j \in X$, if we let $p(j) = i$ we have $p_* H e_j = H' e_i = H' p_* e_j$.  Since the vectors $e_j$ form a basis for $\R^X$, it follows that $p_* H = H' p_*$.

$(ii) \implies (i)$.    Suppose there exists an operator $H' \maps \R^{X'} \to \R^{X'}$ such that $p_* H = H' p_*$.  Choose such an operator; then for any stochastic section $s$ for $p$ we have
\[    p_* Hs = H'p_* s = H' .\]
It follows that $p_* Hs$ is independent of the stochastic section $s$, so $H$ is lumpable with respect to $p$.

Suppose that any, hence all, of conditions $(i), (ii), (iii)$ hold.   Suppose that $H' \maps \R^{X'} \to \R^{X'}$ is an operator with $p_* H = H' p_*$.  Then the argument in the previous paragraph shows that $H' = p_* H s$ for any stochastic section $s$ of $p$.  Thus $H'$ is unique, and by Lemma \ref{lem:new_markov_process} it is infinitesimal stochastic.
\end{proof}

\section{A double category of open Markov processes}
\label{sec:MMark}

In this section we construct a symmetric monoidal double category $\MMark$ with open Markov processes as horizontal 1-cells and morphisms between these as 2-morphisms.   Symmetric monoidal double categories were introduced by Shulman \cite{Shulman} and applied to various examples from engineering by the second author \cite{Courser}.   We refer the reader to those papers for the basic definitions, since they are rather long.

The pieces of the double category $\MMark$ work as follows: 
\begin{enumerate}
\item An object is a finite set. 
\item A vertical 1-morphism $f \maps S \to S'$ is a map between finite sets.
\item A horizontal 1-cell is an open Markov process 
\[
\begin{tikzpicture}[scale=1.5]
\node (D) at (-4,0) {$S$};
\node (E) at (-3,0) {$(X,H)$};
\node (F) at (-2,0) {$T$.};
\path[->,font=\scriptsize,>=angle 90]
(D) edge node[above] {$i$}(E)
(F) edge node[above] {$o$}(E);
\end{tikzpicture}
\]
In other words, it is a pair of injections $S \stackrel{i}{\rightarrow} X \stackrel{o}{\leftarrow} T$
together with a Markov process $H$ on $X$.
\item A 2-morphism is a morphism of open Markov processes
\[
\begin{tikzpicture}[scale=1.5]
\node (D) at (-4.2,-0.5) {$S$};
\node (A) at (-4.2,-2) {$S'$};
\node (B) at (-1.8,-2) {$T'$.};
\node (E) at (-3,-0.5) {$(X,H)$};
\node (F) at (-1.8,-0.5) {$T$};
\node (G) at (-3,-2) {$(X',H')$};
\path[->,font=\scriptsize,>=angle 90]
(D) edge node[above] {$i_1$}(E)
(A) edge node[above] {$i'_1$} (G)
(B) edge node [above]{$o'_1$} (G)
(F) edge node [above]{$o_1$}(E)
(D) edge node [left]{$f$}(A)
(F) edge node [right]{$g$}(B)
(E) edge node[left] {$p$}(G);
\end{tikzpicture}
\]
In other words, it is a triple of maps $f,p,g$ such that these squares are pullbacks:
\[
\begin{tikzpicture}[scale=1.5]
\node (D) at (-4.2,-0.5) {$S$};
\node (A) at (-4.2,-2) {$S'$};
\node (B) at (-1.8,-2) {$T'$,};
\node (E) at (-3,-0.5) {$X$};
\node (F) at (-1.8,-0.5) {$T$};
\node (G) at (-3,-2) {$X'$};
\path[->,font=\scriptsize,>=angle 90]
(D) edge node[above] {$i_1$}(E)
(A) edge node[above] {$i'_1$} (G)
(B) edge node [above]{$o'_1$} (G)
(F) edge node [above]{$o_1$}(E)
(D) edge node [left]{$f$}(A)
(F) edge node [right]{$g$}(B)
(E) edge node[left] {$p$}(G);
\end{tikzpicture}
\]
and $H' p_* = p_* H$.
\end{enumerate}

Composition of vertical 1-morphisms in $\MMark$ is straightforward.  So is vertical 
composition of 2-morphisms, since we can paste two pullback squares and get a new
pullback square.   Composition of horizontal 1-cells is a bit more subtle.   Given 
open Markov processes
\begin{equation}
\label{eq:composable}
\begin{tikzpicture}[scale=1.5]
\node (D) at (-4,0) {$S$};
\node (E) at (-3,0) {$(X,H)$};
\node (F) at (-2,0) {$T,$};
\path[->,font=\scriptsize,>=angle 90]
(D) edge node[above] {$i_1$}(E)
(F) edge node[above] {$o_1$}(E);
\end{tikzpicture}
\qquad
\begin{tikzpicture}[scale=1.5]
\node (D) at (-4,0) {$T$};
\node (E) at (-3,0) {$(Y,G)$};
\node (F) at (-2,0) {$U$};
\path[->,font=\scriptsize,>=angle 90]
(D) edge node[above] {$i_2$}(E)
(F) edge node[above] {$o_2$}(E);
\end{tikzpicture}
\end{equation}
we first compose their underlying cospans using a pushout:
\[
    \xymatrix{
      && X +_T Y \\
      & X \ar[ur]^{j} && Y \ar[ul]_{k} \\
      \quad S\quad \ar[ur]^{i_1} && T \ar[ul]_{o_1} \ar[ur]^{i_2} &&\quad U \quad \ar[ul]_{o_2}
    }
\]
Since monomorphisms are stable under pushout in a topos, the legs of this new cospan are
again injections, as required.   We then define the composite open Markov process to be
\[
\begin{tikzpicture}[scale=1.5]
\node (D) at (-3.5,0) {$S$};
\node (E) at (-2,0) {$(X +_T Y, H \odot G)$};
\node (F) at (-0.5,0) {$U$};
\path[->,font=\scriptsize,>=angle 90]
(D) edge node[above] {$j i_1$}(E)
(F) edge node[above] {$k o_2$}(E);
\end{tikzpicture}
\]
where 
\begin{equation}
\label{eq:odot} 
H \odot G = j_*  H  j^* + k_* G  k^* .
\end{equation}
Here we use both pullbacks and pushforwards along the maps $j$ and $k$, as defined in Def.\ 
\ref{defn:push_and_pull}.  To check that $H \odot G$ is a Markov process on $X +_T Y$ we need to check that $j_* H j^*$ and $k_* G k^*$, and thus their sum, are infinitesimal stochastic:

\begin{lem} 
\label{lem:push-pull}
Suppose that $f \maps X \to Y$ is any map between finite sets.   If $H \maps \R^X \to \R^X$ is infinitesimal stochastic, then $f_* H f^* \maps \R^Y \to \R^Y$ is infinitesimal stochastic.
\end{lem}

\begin{proof}
Using Def.\ \ref{defn:push_and_pull}, we see that the matrix elements of $f^*$ and $f_*$ are
given by
\[               (f^*)_{ji} = (f_*)_{ij} = \left\{ \begin{array}{cc} 1 & f(j) = i \\  0 & \textrm{otherwise} 
                                   \end{array} \right.
\]
for all $i \in Y$, $j \in X$.  Thus, $f_* H f^*$ has matrix entries
\[               (f_* H f^*)_{ii'} = \sum_{j,j': \;  f(j) = i, f(j') = i'} H_{jj'}  .\]
To show that $f_* H f^*$ is infinitesimal stochastic we need to show that its off-diagonal entries are nonnegative and its columns sum to zero.  By the above formula, these follow from the same facts for $H$.
\end{proof}

Another formula for horizontal composition is also useful.  Given the composable open Markov processes in Eq.\ (\ref{eq:composable}) we can take the copairing of the maps $j \maps X \to X +_T Y$ and $k \maps Y \to X +_T Y$ and get a map $\ell \maps X + Y \to X +_T Y$.   Then 
\begin{equation}
\label{eq:odot2}
H \odot G = \ell_* (H \oplus G) \ell^*  
\end{equation}
where $H \oplus G \maps \R^{X + Y} \to \R^{X + Y}$ is the direct sum of the operators $H$ and $G$.   This is easy to check from the definitions.

Horizontal composition of 2-morphisms is even subtler:

\begin{lem}
\label{lem:horizontal_composition}
Suppose that we have horizontally composable 2-morphisms as follows:
\[
\begin{tikzpicture}[scale=1.5]
\node (D) at (-4.2,-0.5) {$S$};
\node (A) at (-4.2,-2) {$S'$};
\node (B) at (-1.8,-2) {$T'$};
\node (E) at (-3,-0.5) {$(X,H)$};
\node (F) at (-1.8,-0.5) {$T$};
\node (D') at (-1.2,-0.5) {$T$};
\node (C) at (-1.2,-2) {$T'$};
\node (H) at (1.2,-2) {$U'$};
\node (E') at (0,-0.5) {$(Y,G)$};
\node (F') at (1.2,-0.5) {$U$};
\node (G) at (-3,-2) {$(X',H')$};
\node (G') at (0,-2) {$(Y',G')$};
\path[->,font=\scriptsize,>=angle 90]
(D) edge node[above] {$i_1$}(E)
(A) edge node[above] {$i'_1$} (G)
(B) edge node [above]{$o'_1$} (G)
(F) edge node [above]{$o_1$}(E)
(D) edge node [left]{$f$}(A)
(F) edge node [right]{$g$}(B)
(E) edge node[left] {$p$}(G)
(D') edge node [above]{$i_2$}(E')
(F') edge node [above]{$o_2$}(E')
(D') edge node [left]{$g$}(C)
(C) edge node [above]{$i'_2$} (G')
(H) edge node [above]{$o'_2$} (G')
(F') edge node[right] {$h$}(H)
(E') edge node[left] {$q$}(G');
\end{tikzpicture}
\]
Then there is a 2-morphism 
\[
\begin{tikzpicture}[scale=1.5]
\node (D) at (-5.5,0) {$S$};
\node (A) at (-5.5,-1.5) {$S'$};
\node (B) at (-0.5,-1.5) {$U'$};
\node (E) at (-3,0) {$(X+_T Y,H \odot G)$};
\node (F) at (-0.5,0) {$U$};
\node (G) at (-3,-1.5) {$(X' +_{T'} Y', H' \odot G')$};
\path[->,font=\scriptsize,>=angle 90]
(D) edge node [above]{$i_3$}(E)
(F) edge node [above]{$o_3$}(E)
(D) edge node [left]{$f$}(A)
(F) edge node [right]{$h$}(B)
(E) edge node[left] {$p+_g q$}(G)
(A) edge node [above]{$i'_3$} (G)
(B) edge node [above]{$o'_3$} (G);
\end{tikzpicture}
\]
whose underlying diagram of finite sets is
\[
\begin{tikzpicture}[scale=1.5]
\node (D) at (-6,0) {$S$};
\node (A) at (-6,-1.5) {$S'$};
\node (E) at (-4.8,0) {$X$};
\node (E') at (-3.5,0) {$X+_T Y$};
\node (F) at (-2.2,0) {$Y$};
\node (F') at (-1,0) {$U$};
\node (G) at (-4.8,-1.5) {$X'$};
\node (G') at (-3.5,-1.5) {$X' +_{T'} Y'$};
\node (H) at (-2.2,-1.5) {$Y'$};
\node (H') at (-1,-1.5) {$U'$,};
\path[->,font=\scriptsize,>=angle 90]
(D) edge node [above]{$i_1$}(E)
(E) edge node [above]{$j$}(E')
(F) edge node [above]{$k$} (E')
(F') edge node [above] {$o_2$} (F)
(D) edge node [left]{$f$}(A)
(E') edge node[left] {$p+_g q$}(G')
(F') edge node [right]{$h$} (H')
(A) edge node [above]{$i'_1$} (G)
(G) edge node [above]{$j'$} (G')
(H) edge node [above]{$k'$} (G')
(H') edge node [above]{$o'_2$} (H);
\end{tikzpicture}
\]
where $j,k,j',k'$ are the canonical maps from $X,Y,X',Y'$, respectively,
to the pushouts $X +_T Y$ and $X' +_{T'} Y'$.
\end{lem}

\begin{proof}
To show that we have defined a 2-morphism, we first check that the squares in the above diagram of finite sets are pullbacks.  Then we show that $(p +_g q)_* (H \odot G) = (H' \odot G') (p +_g q)_*$.

For the first part, it suffices by the symmetry of the situation to consider the left square.   We can write it as a pasting of two smaller squares:
\[
\begin{tikzpicture}[scale=1.5]
\node (D) at (-6,0) {$S$};
\node (A) at (-6,-1.5) {$S'$};
\node (E) at (-4.8,0) {$X$};
\node (E') at (-3.5,0) {$X+_T Y$};
\node (G) at (-4.8,-1.5) {$X'$};
\node (G') at (-3.5,-1.5) {$X' +_{T'} Y'$};
\path[->,font=\scriptsize,>=angle 90]
(D) edge node [above]{$i_1$}(E)
(E) edge node [above]{$j$}(E')
(D) edge node [left]{$f$}(A)
(E) edge node[left] {$p$}(G)
(E') edge node[right] {$p+_g q$}(G')
(A) edge node [above]{$i'_1$} (G)
(G) edge node [above]{$j'$} (G');
\end{tikzpicture}
\]
By assumption the left-hand smaller square is a pullback, so it suffices to prove this for the right-hand one.   For this we use that fact that $\FinSet$ is a topos and thus an adhesive category \cite{LackSobocinski1,LackSobocinski2}, and consider this commutative cube:
\[
\begin{tikzpicture}
\node (T) at (2,3) {$T$};
\node (T') at (2,0) {$T'$};
\node (XTY) at (8,3) {$X +_T Y$};
\node (X) at (6,3.75) {$X$};
\node (Y) at (4,2.25) {$Y$};
\node (X'T'Y') at (8,0) {$X' +_{T'} Y'$};
\node (X') at (6,0.75) {$X'$};
\node (Y') at (4,-0.75) {$Y'$};
\draw [->] (T) edge node[above] {$o_1$} (X);
\draw [->] (T) edge node[above,pos=0.6] {$i_2$} (Y);
\draw [->] (T') edge node[above,pos=0.7] {$o'_1$} (X');
\draw [->] (T') edge node[below,pos=0.4] {$i'_2$} (Y');
\draw [->] (X) edge node[left,pos=0.7] {$p$} (X');
\draw [->] (X) edge node[above] {$j$} (XTY);
\draw [->] (Y) edge node[above,pos=0.4] {$k$} (XTY);
\draw [->] (XTY) edge node[right] {$p +_g q$}(X'T'Y');
\draw [->] (X') edge node[above] {$j'$} (X'T'Y');
\draw [->] (Y') edge node[below] {$k'$} (X'T'Y');

\draw [->] (T) edge[white,line width=4pt] (T');
\draw [->] (T) edge node[left] {$g$} (T');
\draw [->] (Y) edge[white,line width=6pt] (XTY);
\draw [->] (Y) edge (XTY);
\draw [->] (Y) edge[white,line width=6pt] (Y');
\draw [->] (Y) edge node[left,pos=0.4] {$q$} (Y');
\end{tikzpicture}
\]
By assumption the top and bottom faces are pushouts, the two left-hand vertical faces are
pullbacks, and the arrows $o'_1$ and $i_2'$ are monic.   In an adhesive category, this implies that the two right-hand vertical faces are pullbacks as well.   One of these is the square in question.

To show that $(p +_g q)_* (H \odot G) = (H' \odot G') (p +_g q)_*$, we again use the 
above cube.   Because its two right-hand vertical faces commute, we have
\[    (p +_g q)_* j_* = j'_* p_*  \quad \textrm{and} \quad (p +_g q)_* k_* = k'_* q_* \]
so using the definition of $H \odot G$ we obtain
\[ 
\begin{array}{ccl}
(p +_g q)_* (H \odot G) &=& (p +_g q)_* ( j_* H j^* + k_* G  k^*) \\
&=& (p +_g q)_*  j_* H  j^* \; + \; (p +_g q)_*  k_* G  k^* \\
&=& j'_* p_* H j^* \;\; + \; \; k'_* q_* G k^* . 
\end{array}
\]
By assumption we have
\[     p_* H = H' p_*  \quad \textrm{and} \quad q_* G = G' q_*  \]
so we can go a step further, obtaining
\[  (p +_g q)_* (H \odot G) =  j'_* H' p_* j^* \; + \; k'_* G' q_* k^* .\]
Because the two right-hand vertical faces of the cube are pullbacks, Lemma \ref{lem:beck-chevalley} below implies that
\[         p_* j^* = j'^* (p +_g q)_* \quad \textrm{and} \quad 
            q_* k^* = k'^* (p +_g q)_* .
\]
Using these, we obtain
\[ \begin{array}{ccl}
 (p +_g q)_* (H \odot G) 
 &=&  j'_* H'   j'^* (p +_g q)_* \; + \; k'_* G'  k'^* (p +_g q)_* \\
 &=& (j'_* H' j'^* + k'_* G' k'^*) (p +_g q)_* \\
 &=& (H' \odot G') (p +_g q)_* 
\end{array}
\]
completing the proof.
\end{proof}

The following lemma is reminiscent of the Beck--Chevalley condition for adjoint functors:

\begin{lem}
\label{lem:beck-chevalley} 
Given a pullback square in $\FinSet$:
\[
\begin{tikzpicture}[scale=1.5]
\node (D) at (-4,0) {$A$};
\node (E) at (-3,0) {$B$};
\node (G) at (-3,-1) {$D$};
\node (A) at (-4,-1) {$C$};
\path[->,font=\scriptsize,>=angle 90]
(D) edge node [left]{$g$}(A)
(D) edge node [above] {$f$}(E)
(A) edge node[left] [below] {$k$} (G)
(E) edge node[right] {$h$}(G);
\end{tikzpicture}
\]
the following square of linear operators commutes:
\[
\begin{tikzpicture}[scale=1.5]
\node (D) at (-4,0) {$\R^A$};
\node (E) at (-3,0) {$\R^B$};
\node (G) at (-3,-1) {$\R^D$};
\node (A) at (-4,-1) {$\R^C$};
\path[->,font=\scriptsize,>=angle 90]
(D) edge node [left]{$g_*$}(A)
(E) edge node [above] {$f^*$}(D)
(G) edge node[left] [below] {$k^*$} (A)
(E) edge node[right] {$h_*$}(G);
\end{tikzpicture}
\]
\end{lem}

\begin{proof}
Choose $v \in \R^B$ and $c \in C$.  Then
\[ 
\begin{array}{ccl}
(g_* f^* (v))(c) &=& \displaystyle{ \sum_{a: g(a) = c} v(f(a)), }\\ \\
(k^* h_* (v))(c) &=& \displaystyle{ \sum_{b:  h(b) = k(c)}  v(b) ,  }
\end{array}
\]
so to show $g_* f^* = k^* h_*$ it suffices to show that $f$ restricts to a bijection
\[      f \maps \{ a \in A :  g(a) = c \} \stackrel{\sim}{\longrightarrow} \{ b \in B: h(b) = k(c) \} .\] 
On the one hand, if $a \in A$ has $g(a) = c$ then $b = f(a)$ has
$h(b) = h(f(a)) = k(g(a)) = k(c)$, so the above map is well-defined.  On the other hand, 
if $b \in B$ has $h(b) = k(c)$, then by the definition of pullback there exists a unique $a \in A$
such that $f(a) = b$ and $g(a) = c$, so the above map is a bijection.
\end{proof}

\begin{thm}
There exists a double category $\MMark$ as defined above.
\end{thm}

\begin{proof}
Let $\MMark_0$, the `category of objects', consist of finite sets and functions.  Let 
$\MMark_1$ the `category of arrows', consist of open Markov processes and morphisms between these:
\[
\begin{tikzpicture}[scale=1.5]
\node (D) at (-4.2,-0.5) {$S$};
\node (A) at (-4.2,-2) {$S'$};
\node (B) at (-1.8,-2) {$T'$.};
\node (E) at (-3,-0.5) {$(X,H)$};
\node (F) at (-1.8,-0.5) {$T$};
\node (G) at (-3,-2) {$(X',H')$};
\path[->,font=\scriptsize,>=angle 90]
(D) edge node[above] {$i_1$}(E)
(A) edge node[above] {$i'_1$} (G)
(B) edge node [above]{$o'_1$} (G)
(F) edge node [above]{$o_1$}(E)
(D) edge node [left]{$f$}(A)
(F) edge node [right]{$g$}(B)
(E) edge node[left] {$p$}(G);
\end{tikzpicture}
\]
To make $\MMark$ into a double category we need to specify the identity-assigning functor 
\[   u \maps \MMark_0 \to \MMark_1, \]
the source and target functors
\[   s,t \maps \MMark_1 \to \MMark_0, \]
and the composition functor
\[ \odot \maps \MMark_1 \times_{\MMark_0} \MMark_1 \to \MMark_1 .\] 
These are given as follows.

For a finite set $S$, $u(S)$ is given by
\[
\begin{tikzpicture}[scale=1.5]
\node (D) at (-4,0) {$S$};
\node (E) at (-3,0) {$(S,0_S)$};
\node (F) at (-2,0) {$S$};
\path[->,font=\scriptsize,>=angle 90]
(D) edge node[above] {$1_S$}(E)
(F) edge node[above] {$1_S$}(E);
\end{tikzpicture}
\]
where $0_S$ is the zero operator from $\R^S$ to $\R^S$.   For a map 
$f \maps S \to S'$ between finite sets, $u(f)$ is given by
\[
\begin{tikzpicture}[scale=1.5]
\node (D) at (-4,0.5) {$S$};
\node (E) at (-3,0.5) {$(S,0_S)$};
\node (F) at (-2,0.5) {$S$};
\node (A) at (-4,-1) {$S'$};
\node (B) at (-2,-1) {$S'$};
\node (G) at (-3,-1) {$(S',0_{S'})$};
\path[->,font=\scriptsize,>=angle 90]
(D) edge node {$$}(E)
(F) edge node {$$}(E)
(D) edge node [left]{$f$}(A)
(F) edge node [right]{$f$}(B)
(A) edge node {$$} (G)
(B) edge node {$$} (G)
(E) edge node[left] {$f$}(G);
\end{tikzpicture}
\]
The source and target functors $s$ and $t$ map a Markov process $S \stackrel{i}{\rightarrow} (X,H) \stackrel{o}{\leftarrow} T$ to $S$ and $T$, respectively, and they map a morphism of open Markov processes 
\[
\begin{tikzpicture}[scale=1.5]
\node (D) at (-4.2,-0.5) {$S$};
\node (A) at (-4.2,-2) {$S'$};
\node (B) at (-1.8,-2) {$T'$};
\node (E) at (-3,-0.5) {$(X,H)$};
\node (F) at (-1.8,-0.5) {$T$};
\node (G) at (-3,-2) {$(X',H')$};
\path[->,font=\scriptsize,>=angle 90]
(D) edge node[above] {$i_1$}(E)
(A) edge node[above] {$i'_1$} (G)
(B) edge node [above]{$o'_1$} (G)
(F) edge node [above]{$o_1$}(E)
(D) edge node [left]{$f$}(A)
(F) edge node [right]{$g$}(B)
(E) edge node[left] {$p$}(G);
\end{tikzpicture}
\]
to $f \maps S \to S'$ and $g \maps T \to T'$, respectively.   The composition functor $\odot$ maps the pair of open Markov processes 
\[
\begin{tikzpicture}[scale=1.5]
\node (D) at (-4.2,-0.5) {$S$};
\node (E) at (-3,-0.5) {$(X,H)$};
\node (F) at (-1.8,-0.5) {$T$};
\node (D') at (-1.2,-0.5) {$T$};
\node (E') at (0,-0.5) {$(Y,G)$};
\node (F') at (1.2,-0.5) {$U$};
\path[->,font=\scriptsize,>=angle 90]
(D) edge node[above] {$i_1$}(E)
(F) edge node[above] {$o_1$}(E)
(D') edge node[above] {$i_2$}(E')
(F') edge node [above]{$o_2$}(E');
\end{tikzpicture}
\]
to their composite
\[
\begin{tikzpicture}[scale=1.5]
\node (D) at (-3.5,0) {$S$};
\node (E) at (-2,0) {$(X +_T Y, H \odot G)$};
\node (F) at (-0.5,0) {$U$};
\path[->,font=\scriptsize,>=angle 90]
(D) edge node[above] {$j i_1$}(E)
(F) edge node[above] {$k o_2$}(E);
\end{tikzpicture}
\]
defined as in Eq.\ (\ref{eq:odot}), and it maps the pair of morphisms of open
Markov processes
\[
\begin{tikzpicture}[scale=1.5]
\node (D) at (-4.2,-0.5) {$S$};
\node (A) at (-4.2,-2) {$S'$};
\node (B) at (-1.8,-2) {$T'$};
\node (E) at (-3,-0.5) {$(X,H)$};
\node (F) at (-1.8,-0.5) {$T$};
\node (D') at (-1.2,-0.5) {$T$};
\node (C) at (-1.2,-2) {$T'$};
\node (H) at (1.2,-2) {$U'$};
\node (E') at (0,-0.5) {$(Y,G)$};
\node (F') at (1.2,-0.5) {$U$};
\node (G) at (-3,-2) {$(X',H')$};
\node (G') at (0,-2) {$(Y',G')$};
\path[->,font=\scriptsize,>=angle 90]
(D) edge node[above] {$i_1$}(E)
(A) edge node[above] {$i'_1$} (G)
(B) edge node [above]{$o'_1$} (G)
(F) edge node [above]{$o_1$}(E)
(D) edge node [left]{$f$}(A)
(F) edge node [right]{$g$}(B)
(E) edge node[left] {$p$}(G)
(D') edge node [above]{$i_2$}(E')
(F') edge node [above]{$o_2$}(E')
(D') edge node [left]{$g$}(C)
(C) edge node [above]{$i'_2$} (G')
(H) edge node [above]{$o'_2$} (G')
(F') edge node[right] {$h$}(H)
(E') edge node[left] {$q$}(G');
\end{tikzpicture}
\]
to their horizontal composite as defined as in Lemma \ref{lem:horizontal_composition}.

It is easy to check that $u,s$ and $t$ are functors.  To prove that $\odot$ is a functor, the main thing we need to check is the interchange law.    Suppose we have four morphisms of open Markov processes as follows:
\[
\begin{tikzpicture}[scale=1.5]
\node (D) at (-4.2,-0.5) {$S$};
\node (A) at (-4.2,-2) {$S'$};
\node (B) at (-1.8,-2) {$T'$};
\node (E) at (-3,-0.5) {$(X,H)$};
\node (F) at (-1.8,-0.5) {$T$};
\node (D') at (-1.2,-0.5) {$T$};
\node (C) at (-1.2,-2) {$T'$};
\node (H) at (1.2,-2) {$U'$};
\node (E') at (0,-0.5) {$(Y,G)$};
\node (F') at (1.2,-0.5) {$U$};
\node (G) at (-3,-2) {$(X',H')$};
\node (G') at (0,-2) {$(Y',G')$};
\node (D'') at (-4.2,-2.5) {$S'$};
\node (I) at (-4.2,-4) {$S''$};
\node (J) at (-1.8,-4) {$T''$};
\node (E'') at (-3,-2.5) {$(X',H')$};
\node (F'') at (-1.8,-2.5) {$T'$};
\node (D''') at (-1.2,-2.5) {$T'$};
\node (E''') at (0,-2.5) {$(Y',G')$};
\node (F''') at (1.2,-2.5) {$U'$};
\node (K) at (-1.2,-4) {$T''$};
\node (L) at (1.2,-4) {$U''$};
\node (G'') at (-3,-4) {$(X'',H'')$};
\node (G''') at (0,-4) {$(Y'',G'')$};
\path[->,font=\scriptsize,>=angle 90]
(D) edge node {$$}(E)
(A) edge node {$$} (G)
(B) edge node {$$} (G)
(F) edge node {$$}(E)
(D) edge node [left]{$f$}(A)
(F) edge node [right]{$g$}(B)
(E) edge node[left] {$p$}(G)
(D') edge node {$$}(E')
(F') edge node {$$}(E')
(D') edge node [left]{$g$}(C)
(C) edge node {$$} (G')
(H) edge node {$$} (G')
(F') edge node[right] {$h$}(H)
(E') edge node[left] {$q$}(G')
(D'') edge node {$$}(E'')
(F'') edge node {$$}(E'')
(D'') edge node [left]{$f'$}(I)
(F'') edge node [right]{$g'$}(J)
(I) edge node {$$} (G'')
(J) edge node {$$} (G'')
(E'') edge node[left] {$p'$}(G'')
(D''') edge node {$$}(E''')
(F''') edge node {$$}(E''')
(D''') edge node [left]{$g'$}(K)
(F''') edge node [right]{$h'$}(L)
(K) edge node {$$} (G''')
(L) edge node {$$} (G''')
(E''') edge node[left] {$q'$}(G''');
\end{tikzpicture}
\]
Composing horizontally gives
\[
\begin{tikzpicture}[scale=1.5]
\node (D) at (-5.5,0) {$S$};
\node (A) at (-5.5,-1.5) {$S'$};
\node (B) at (-0.5,-1.5) {$U'$};
\node (C) at (-5.5,-4) {$S''$};
\node (H) at (-.5,-4) {$U''$,};
\node (E) at (-3,0) {$(X+_T Y,H \odot G)$};
\node (F) at (-0.5,0) {$U$};
\node (G) at (-3,-1.5) {$(X' +_{T'} Y', H' \odot G')$};
\node (D'') at (-5.5,-2.5) {$S'$};
\node (E'') at (-3,-2.5) {$(X' +_{T'} Y', H' \odot G')$};
\node (F'') at (-0.5,-2.5) {$U'$};
\node (G'') at (-3,-4) {$(X'' +_{T''} Y'', H'' \odot G'')$};
\path[->,font=\scriptsize,>=angle 90]
(D) edge node {$$}(E)
(F) edge node {$$}(E)
(D) edge node [left]{$f$}(A)
(F) edge node [right]{$h$}(B)
(E) edge node[left] {$p+_g q$}(G)
(A) edge node {$$} (G)
(B) edge node {$$} (G)
(D'') edge node {$$}(E'')
(F'') edge node {$$}(E'')
(C) edge node {$$} (G'')
(H) edge node {$$} (G'')
(D'') edge node [left]{$f'$}(C)
(F'') edge node [right]{$h'$}(H)
(E'') edge node[left] {$p'+_{g'}q'$}(G'');
\end{tikzpicture}
\]
and then composing vertically gives
\[
\begin{tikzpicture}[scale=1.5]
\node (D) at (-5.5,0.5) {$S$};
\node (A) at (-5.5,-1) {$S''$};
\node (B) at (-0.5,-1) {$U''$.};
\node (E) at (-3,.5) {$(X+_T Y,H \odot G)$};
\node (F) at (-0.5,0.5) {$U$};
\node (G) at (-3,-1) {$(X'' +_{T''} Y'', H'' \odot G'')$};
\path[->,font=\scriptsize,>=angle 90]
(D) edge node {$$}(E)
(F) edge node {$$}(E)
(D) edge node [left]{$f' \circ f$}(A)
(F) edge node [right]{$h' \circ h$}(B)
(A) edge node {$$} (G)
(B) edge node {$$} (G)
(E) edge node[left] {$(p' +_{g'}q') \circ (p+_g q)$}(G);
\end{tikzpicture}
\]
Composing vertically gives
\[
\begin{tikzpicture}[scale=1.5]
\node (D) at (-4.75,0.5) {$S$};
\node (E) at (-3.5,0.5) {$(X,H)$};
\node (F) at (-2.25,0.5) {$T$};
\node (D') at (-1,0.5) {$T$};
\node (E') at (0.25,.5) {$(Y,G)$};
\node (F') at (1.5,0.5) {$U$};
\node (A) at (-4.75,-1) {$S''$};
\node (B) at (-2.25,-1) {$T''$};
\node (C) at (-1,-1) {$T''$};
\node (H) at (1.5,-1) {$U''$,};
\node (G) at (-3.5,-1) {$(X'',H'')$};
\node (G') at (0.25,-1) {$(Y'',G'')$};
\path[->,font=\scriptsize,>=angle 90]
(D) edge node {$$}(E)
(F) edge node {$$}(E)
(D) edge node [left]{$f' \circ f$}(A)
(F) edge node [right]{$g' \circ g$}(B)
(A) edge node {$$} (G)
(B) edge node {$$} (G)
(E) edge node[left] {$p' \circ p$}(G)
(D') edge node {$$}(E')
(F') edge node {$$}(E')
(D') edge node [left]{$g' \circ g$}(C)
(F') edge node [right]{$h' \circ h$}(H)
(C) edge node {$$} (G')
(H) edge node {$$} (G')
(E') edge node[right] {$q' \circ q$}(G');
\end{tikzpicture}
\]
and then composing horizontally gives
\[
\begin{tikzpicture}[scale=1.5]
\node (D) at (-5.5,0.5) {$S$};
\node (A) at (-5.5,-1) {$S''$};
\node (B) at (-0.5,-1) {$U''$.};
\node (E) at (-3,0.5) {$(X+_T Y,H \odot G)$};
\node (F) at (-0.5,0.5) {$U$};
\node (G) at (-3,-1) {$(X'' +_{T''} Y'', H'' \odot G'')$};
\path[->,font=\scriptsize,>=angle 90]
(D) edge node {$$}(E)
(F) edge node {$$}(E)
(D) edge node [left]{$f' \circ f$}(A)
(F) edge node [right]{$h' \circ h$}(B)
(A) edge node {$$} (G)
(B) edge node {$$} (G)
(E) edge node[left] {$(p' \circ p) +_{(g' \circ g)} (q' \circ q)$}(G);
\end{tikzpicture}
\]
The only apparent difference between the two results is the map in the middle: one has $(p' +_{g'}q') \circ (p+_g q)$ while the other has $(p' \circ p) +_{(g' \circ g)} (q' \circ q)$.  But these are in fact the same map, so the interchange law holds. 

The functors $u, s, t$ and $\circ$ obey the necessary relations
\[      s u = 1 = t  u  \]
and the relations saying that the source and target of a composite behave as they should.
Lastly, we have three natural isomorphisms: the associator, left unitor, and right unitor,
which arise from the corresponding natural isomorphisms for the double category of finite sets, functions, cospans of finite sets, and maps of cospans.   The triangle and pentagon equations 
hold in $\MMark$ because they do in this simpler double category \cite{Courser}.
\end{proof}

Next we give $\MMark$ a symmetric monoidal structure.  We call the tensor product `addition'.   Given objects $S,S' \in \MMark_0$ we define their sum $S+S'$ using a chosen coproduct in $\FinSet$.   The unit for this tensor product in $\MMark_0$ is the empty set.  We can similarly define the sum of morphisms in $\MMark_0$, since given maps $f \maps S \to T$ and $f' \maps S' \to T'$ there is a natural map $f + f' \maps S + S' \to T + T'$.   Given two objects in $\MMark_1$:
\[
\begin{tikzpicture}[scale=1.5];
\node (D) at (-4.4,0.5) {$S_1$};
\node (E) at (-3.2,0.5) {$(X_1,H_1)$};
\node (F) at (-2,0.5) {$T_1$};
\node (A) at (-1,0.5) {$S_2$};
\node (B) at (0.2,0.5) {$(X_2,H_2)$};
\node (C) at (1.4,0.5) {$T_2$};
\path[->,font=\scriptsize,>=angle 90]
(D) edge node [above] {$i_1$}(E)
(F) edge node [above] {$o_1$}(E)
(A) edge node [above] {$i_2$} (B)
(C) edge node [above] {$o_2$} (B);
\end{tikzpicture}
\]
we define their sum to be
\[
\begin{tikzpicture}[scale=1.5];
\node (D) at (-5.5,0.5) {$S_1+S_2$};
\node (E) at (-3,0.5) {$(X_1 + X_2,H_1 \oplus H_2)$};
\node (F) at (-0.5,0.5) {$T_1+T_2$};
\path[->,font=\scriptsize,>=angle 90]
(D) edge node [above] {$i_1 + i_2$}(E)
(F) edge node [above] {$o_1 + o_2$}(E);
\end{tikzpicture}
\]
where $H_1 \oplus H_2 \maps \R^{X_1 + X_2} \to \R^{X_1 + X_2}$ is the direct sum of the operators $H_1$ and $H_2$.  The unit for this tensor product in $\MMark_1$ is $\emptyset \rightarrow (\emptyset, 0_\emptyset) \leftarrow \emptyset$ where $0_\emptyset \maps \R^\emptyset \to \R^\emptyset$ is the zero operator.  Finally, given two morphisms in $\MMark_1$:
\[
\begin{tikzpicture}[scale=1.5]
\node (D) at (-4.75,0.5) {$S_1$};
\node (A) at (-4.75,-1) {$S_1'$};
\node (B) at (-2.25,-1) {$T_1'$};
\node (C) at (-1.25,-1) {$S_2'$};
\node (H) at (1.25,-1) {$T_2'$};
\node (E) at (-3.5,0.5) {$(X_1,H_1)$};
\node (F) at (-2.25,0.5) {$T_1$};
\node (D') at (-1.25,0.5) {$S_2$};
\node (E') at (0,.5) {$(X_2,H_2)$};
\node (F') at (1.25,0.5) {$T_2$};
\node (G) at (-3.5,-1) {$(X'_1,H'_1)$};
\node (G') at (0,-1) {$(X'_2,H'_2)$};
\path[->,font=\scriptsize,>=angle 90]
(D) edge node [above] {$i_1$}(E)
(F) edge node [above] {$o_1$}(E)
(D) edge node [left]{$f_1$}(A)
(F) edge node [right]{$g_1$}(B)
(A) edge node [above] {$i'_1$} (G)
(B) edge node [above] {$o'_1$} (G)
(E) edge node[left] {$p_1$}(G)
(D') edge node [above] {$i_2$}(E')
(F') edge node [above] {$o_2$}(E')
(D') edge node [left]{$f_2$}(C)
(F') edge node [right]{$g_2$}(H)
(C) edge node [above] {$i'_2$} (G')
(H) edge node [above] {$o'_2$} (G')
(E') edge node[right]{$p_2$}(G');
\end{tikzpicture}
\]
we define their sum to be
\[
\begin{tikzpicture}[scale=1.5]
\node (D) at (-6,0.5) {$S_1+S_2$};
\node (A) at (-6,-1) {$S_1' + S_2'$};
\node (B) at (0,-1) {$T_1' + T_2'.$};
\node (E) at (-3,.5) {$(X_1+X_2,H_1 \oplus H_2)$};
\node (F) at (0,0.5) {$T_1+T_2$};
\node (G) at (-3,-1) {$(X'_1+X'_2, H'_1\oplus H'_2)$};
\path[->,font=\scriptsize,>=angle 90]
(D) edge node [above] {$i_1 + i_2$}(E)
(F) edge node [above] {$o_1 + o_2$}(E)
(D) edge node [left]{$f_1 + f_2$}(A)
(F) edge node [right]{$g_1+g_2$}(B)
(A) edge node [above] {$i'_1 + i'_2$} (G)
(B) edge node [above] {$o'_1 + o'_2$} (G)
(E) edge node[left] {$p_1 + p_2$}(G);
\end{tikzpicture}
\]
We complete the description of $\MMark$ as a symmetric monoidal double category
in the proof of this theorem:

\begin{thm}
\label{thm:MMark_symmetric_monoidal}
The double category $\MMark$ can be given a symmetric monoidal structure with the above properties.
\end{thm}

\begin{proof}
First we complete the description of $\MMark_0$ and $\MMark_1$ as symmetric monoidal categories. The symmetric monoidal category $\MMark_0$ is just the category of finite sets with a chosen  coproduct of each pair of finite sets providing the symmetric monoidal structure.   We have described the tensor product in $\MMark_1$, which we call `addition', so now we need to introduce the associator, unitors, and braiding, and check that they make $\MMark_1$ into a symmetric monoidal category. 

Given three objects in $\MMark_1$
\[
\begin{tikzpicture}[scale=1.5];
\node (D) at (-4,0.5) {$S_1$};
\node (E) at (-3,0.5) {$(X_1,H_1)$};
\node (F) at (-2,0.5) {$T_1$};
\node (A) at (-1,0.5) {$S_2$};
\node (B) at (0,0.5) {$(X_2,H_2)$};
\node (C) at (1,0.5) {$T_2$};
\node (G) at (2,0.5) {$S_3$};
\node (H) at (3,0.5) {$(X_3,H_3)$};
\node (I) at (4,0.5) {$T_3$};
\path[->,font=\scriptsize,>=angle 90]
(D) edge node {$$}(E)
(F) edge node {$$}(E)
(A) edge node {$$} (B)
(C) edge node {$$} (B)
(G) edge node {$$} (H)
(I) edge node {$$} (H);
\end{tikzpicture}
\]
tensoring the first two and then the third results in
\[
\begin{tikzpicture}[scale=1.5];
\node (D) at (-6,0.5) {$(S_1+S_2)+S_3$};
\node (E) at (-3,0.5) {$((X_1+X_2)+X_3,(H_1 \oplus H_2) \oplus H_3)$};
\node (F) at (0,0.5) {$(T_1+T_2)+T_3$};
\path[->,font=\scriptsize,>=angle 90]
(D) edge node {$$}(E)
(F) edge node {$$}(E);
\end{tikzpicture}
\]
whereas tensoring the last two and then the first results in
\[
\begin{tikzpicture}[scale=1.5];
\node (D) at (-6,0.5) {$S_1+(S_2+S_3)$};
\node (E) at (-3,0.5) {$(X_1+(X_2+X_3),H_1 \oplus (H_2 \oplus H_3))$};
\node (F) at (0,0.5) {$T_1+(T_2+T_3)$.};
\path[->,font=\scriptsize,>=angle 90]
(D) edge node {$$}(E)
(F) edge node {$$}(E);
\end{tikzpicture}
\]
The associator for $\MMark_1$ is then given as follows:
\[
\begin{tikzpicture}[scale=1.5]
\node (D) at (-4,0.5) {$(S_1+S_2)+S_3$};
\node (E) at (-1,.5) {$((X_1+X_2)+X_3,(H_1 \oplus H_2) \oplus H_3)$};
\node (F) at (2,0.5) {$(T_1+T_2)+T_3$};
\node (G) at (-1,-1.0) {$(X_1+(X_2+X_3),H_1 \oplus (H_2 \oplus H_3))$};
\node (A) at (-4,-1.0) {$S_1+(S_2+S_3)$};
\node (B) at (2,-1.0) {$T_1+(T_2+T_3)$};
\path[->,font=\scriptsize,>=angle 90]
(D) edge node [left]{$a$}(A)
(F) edge node [right]{$a$}(B)
(D) edge node {$$}(E)
(F) edge node {$$}(E)
(A) edge node[left] {$$} (G)
(B) edge node[left] {$$} (G)
(E) edge node[left] {$a$}(G);
\end{tikzpicture}
\]
where $a$ is the associator in $(\FinSet, +)$.  If we abbreviate an object $S \rightarrow (X,H) \leftarrow T$ of $\MMark_1$ as $(X,H)$, and denote the associator for $\MMark_1$ as $\alpha$, the pentagon identity says that this diagram commutes:
\[
\begin{tikzpicture}[scale=1.5];
\node (D) at (-5.5,-0.5) {$(((X_1,H_1) \oplus (X_2,H_2)) \oplus (X_3,H_3)) \oplus (X_4,H_4)$};
\node (E) at (-3,.5) {$((X_1,H_1) \oplus (X_2,H_2)) \oplus ((X_3,H_3) \oplus (X_4,H_4))$};
\node (F) at (-0.5,-0.5) {$(X_1,H_1) \oplus ((X_2,H_2) \oplus ((X_3,H_3) \oplus (X_4,H_4)))$};
\node (A) at (-0.5,-1.5) {$(X_1,H_1) \oplus (((X_2,H_2) \oplus (X_3,H_3)) \oplus (X_4,H_4))$};
\node (B) at (-5.5,-1.5) {$((X_1,H_1) \oplus ((X_2,H_2) \oplus (X_3,H_3))) \oplus (X_4,H_4)$};
\path[->,font=\scriptsize,>=angle 90]
(D) edge node[above] {$\alpha$}(E)
(E) edge node [above]{$\alpha$}(F)
(D) edge node[below,left] {$\alpha \oplus 1_{(X_4,H_4)}$} (B)
(B) edge node [above]{$\alpha$} (A)
(A) edge node [right]{$1_{(X_1,H_1)} \oplus \alpha$} (F);

\end{tikzpicture}
\]
which is clearly true.  Recall that the monoidal unit for $\MMark_1$ is given by  $\emptyset \rightarrow (\emptyset, 0_\emptyset) \leftarrow \emptyset$. The left and right unitors for $\MMark_1$, denoted $\lambda$ and $\rho$, are given respectively by the following 2-morphisms:
\[
\begin{tikzpicture}[scale=1.5]
\node (D) at (-5,0.5) {$\emptyset + S$};
\node (A) at (-5,-1.0) {$S$};
\node (B) at (-2,-1.0) {$T$};
\node (C) at (-0.5,-1.0) {$S$};
\node (H) at (2.5,-1.0) {$T$};
\node (E) at (-3.5,.5) {$(\emptyset + X,0_\emptyset \oplus H)$};
\node (F) at (-2,0.5) {$\emptyset + T$};
\node (D') at (-0.5,0.5) {$S + \emptyset$};
\node (E') at (1,.5) {$(X + \emptyset,H \oplus 0_\emptyset)$};
\node (F') at (2.5,0.5) {$T + \emptyset$};
\node (G) at (-3.5,-1.0) {$(X,H)$};
\node (G') at (1,-1.0) {$(X,H)$};
\path[->,font=\scriptsize,>=angle 90]
(D) edge node {$$}(E)
(F) edge node {$$}(E)
(D) edge node [left]{$\ell$}(A)
(F) edge node [right]{$\ell$}(B)
(A) edge node {$$} (G)
(B) edge node {$$} (G)
(E) edge node[left] {$\ell$}(G)
(D') edge node {$$}(E')
(F') edge node {$$}(E')
(D') edge node [left]{$r$}(C)
(F') edge node [right]{$r$}(H)
(C) edge node {$$} (G')
(H) edge node {$$} (G')
(E') edge node[left]{$r$}(G');
\end{tikzpicture}
\]
where $\ell$ and $r$ are the left and right unitors in $\FinSet$.
The left and right unitors and associator for $\MMark_1$ satisfy the triangle identity:
\[
\begin{tikzpicture}[scale=1.5];
\node (D) at (-2.5,-0.5) {$((X,H) \oplus (\emptyset,0_\emptyset)) \oplus (Y,G)$};
\node (E) at (-1,0.5) {$(X,H) \oplus (Y,G)$};
\node (F) at (0.5,-0.5) {$(X,H) \oplus ((\emptyset,0_\emptyset) \oplus (Y,G))$.};
\path[->,font=\scriptsize,>=angle 90]
(D) edge node[above] {$\rho \oplus 1\;$}(E)
(F) edge node [above]{$\; 1 \oplus \lambda$}(E)
(D) edge node [above]{$\alpha$} (F);
\end{tikzpicture}
\]
The braiding in $\MMark_1$ is given as follows:
\[
\begin{tikzpicture}[scale=1.5]
\node (D) at (-6,0.5) {$S_1+S_2$};
\node (A) at (-6,-1.0) {$S_2+S_1$};
\node (B) at (-1,-1.0) {$T_2+T_1$};
\node (E) at (-3.5,.5) {$(X_1,H_1) \oplus (X_2,H_2)$};
\node (F) at (-1,0.5) {$T_1+T_2$};
\node (G) at (-3.5,-1.0) {$(X_2,H_2) \oplus (X_1,H_1)$};
\path[->,font=\scriptsize,>=angle 90]
(D) edge node {$$}(E)
(F) edge node {$$}(E)
(D) edge node [left]{$b_{S_1,S_2}$}(A)
(F) edge node [left]{$b_{T_1,T_2}$}(B)
(A) edge node {$$} (G)
(B) edge node {$$} (G)
(E) edge node[left] {$b_{X_1,X_2}$}(G);
\end{tikzpicture}
\] 
where $b$ is the braiding in $(\FinSet,+)$.   It is easy to check that the braiding in $\MMark_1$ is its own inverse and obeys the hexagon identity, making $\MMark_1$ into a symmetric monoidal category.

The source and target functors $s,t \maps \MMark_1 \to \MMark_0$ are strict symmetric monoidal functors, as required.  To make $\MMark$ into a symmetric monoidal double category we must also give it two other pieces of structure.  One, called $\chi$, says how the composition of horizontal 1-cells interacts with the tensor product in the category of arrows.  The other, called $\mu$, says how the identity-assigning functor $u$ relates the tensor product in the category of objects to the tensor product in the category of arrows. We now define these two isomorphisms.

Given horizontal 1-cells 
\[
\begin{tikzpicture}[scale=1.5];
\node (D) at (-4,-1.5) {$S_1$};
\node (E) at (-3,-1.5) {$(X_1,H_1)$};
\node (F) at (-2,-1.5) {$T_1$};
\node (A) at (-1,-1.5) {$T_1$};
\node (B) at (0,-1.5) {$(Y_1,G_1)$};
\node (C) at (1,-1.5) {$U_1$};
\node (D') at (-4,-2) {$S_2$};
\node (E') at (-3,-2) {$(X_2,H_2)$};
\node (F') at (-2,-2) {$T_2$};
\node (A') at (-1,-2) {$T_2$};
\node (B') at (0,-2) {$(Y_2,G_2)$};
\node (C') at (1,-2) {$U_2$};
\path[->,font=\scriptsize,>=angle 90]
(D) edge node [above] {$$}(E)
(A) edge node [above] {$$} (B)
(C) edge node [above] {$$} (B)
(F) edge node [above] {$$}(E)
(D') edge node [above] {$$}(E')
(A') edge node [above] {$$} (B')
(C') edge node [above] {$$} (B')
(F') edge node [above] {$$}(E');
\end{tikzpicture}
\]
the horizontal composites of the top two and the bottom two are given, respectively, by
\[
\begin{tikzpicture}[scale=1.5];
\node (D) at (-3,-0.5) {$S_1$};
\node (E) at (-1.5,-0.5) {$(X_1 +_{T_1} Y_1,H_1 \odot G_1)$};
\node (F) at (0,-0.5) {$U_1$};
\node (A) at (1,-0.5) {$S_2$};
\node (B) at (2.5,-0.5) {$(X_2 +_{T_2} Y_2,H_2 \odot G_2)$};
\node (C) at (4,-0.5) {$U_2$.};
\path[->,font=\scriptsize,>=angle 90]
(D) edge node [above] {$$}(E)
(A) edge node [above] {$$} (B)
(C) edge node [above] {$$} (B)
(F) edge node [above] {$$}(E);
\end{tikzpicture}
\]
`Adding' the left two and right two, respectively, we obtain
\[
\begin{tikzpicture}[scale=1.5];
\node (D) at (-3.5,-0.5) {$S_1+S_2$};
\node (E) at (-1.875,-0.5) {$(X_1 + X_2, H_1 \oplus H_2)$};
\node (F) at (-0.25,-0.5) {$T_1+T_2$};
\node (A) at (1,-0.5) {$T_1+T_2$};
\node (B) at (2.625,-0.5) {$(Y_1 + Y_2, G_1 \oplus G_2)$};
\node (C) at (4.25,-0.5) {$U_1+U_2$.};
\path[->,font=\scriptsize,>=angle 90]
(D) edge node [above] {$$}(E)
(A) edge node [above] {$$} (B)
(C) edge node [above] {$$} (B)
(F) edge node [above] {$$}(E);
\end{tikzpicture}
\]
Thus the sum of the horizontal composites is 
\[
\begin{tikzpicture}[scale=1.5];
\node (D) at (-7.5,0.5) {$S_1+S_2$};
\node (E) at (-4,.5) {$((X_1 +_{T_1} Y_1) + (X_2 +_{T_2} Y_2),(H_1 \odot G_1) \oplus (H_2 \odot G_2))$};
\node (F) at (-0.5,0.5) {$U_1 + U_2$};
\path[->,font=\scriptsize,>=angle 90]
(D) edge node [above] {$$}(E)
(F) edge node [above] {$$} (E);
\end{tikzpicture}
\]
while the horizontal composite of the sums is
\[
\begin{tikzpicture}[scale=1.5];
\node (D) at (-7.5,0.5) {$S_1+S_2$};
\node (E) at (-4,0.5) {$((X_1+X_2)+_{T_1 + T_2} (Y_1+Y_2),(H_1 \oplus H_2) \odot (G_1 \oplus G_2))$};
\node (F) at (-0.5,0.5) {$U_1 + U_2$.};
\path[->,font=\scriptsize,>=angle 90]
(D) edge node [above] {$$}(E)
(F) edge node [above] {$$} (E);
\end{tikzpicture}
\]
The required globular 2-isomorphism $\chi$ between these is
\[
\begin{tikzpicture}[scale=1.5]
\node (D) at (-7,0.5) {$S_1+S_2$};
\node (A) at (-7,-1.0) {$S_1+S_2$};
\node (B) at (0,-1.0) {$U_1+U_2$};
\node (E) at (-3.5,.5) {$((X_1,H_1) \odot (Y_1,G_1)) \oplus ((X_2,H_2) \odot (Y_2,G_2))$};
\node (F) at (0,0.5) {$U_1+U_2$};
\node (G) at (-3.5,-1.0) {$ ((X_1,H_1) \oplus (X_2,H_2)) \odot ((Y_1,G_1) \oplus (Y_2,G_2))$};
\path[->,font=\scriptsize,>=angle 90]
(D) edge node {$$}(E)
(F) edge node {$$}(E)
(D) edge node [left]{$1_{S_1 + S_2}$}(A)
(F) edge node [right]{$1_{U_1 + U_2}$}(B)
(A) edge node {$$} (G)
(B) edge node {$$} (G)
(E) edge node[left] {$\hat{\chi}$}(G);
\end{tikzpicture}
\] 
where $\hat{\chi}$ is the bijection 
\[ \hat{\chi} \maps (X_1 +_{T_1} Y_1) + (X_2 +_{T_2} Y_2) \to 
(X_1 + X_2) +_{T_1 + T_2} (Y_1 + Y_2)\]
obtained from taking the colimit of the diagram
\begin{center}
\begin{tikzpicture}[->,>=stealth',node distance=1.1cm, auto]
 \node(A1) {$S_{1}$};
 \node(C1) [above of=A1,right of=A1] {$X_{1}$};
 \node(B1) [below of=C1,right of=C1] {$T_{1}$};
 \node(D1) [above of=B1,right of=B1] {$Y_{1}$};
 \node(E1) [below of=D1,right of=D1] {$U_{1}$};
 \node(A2) [right of=E1] {$S_{2}$};
 \node(C2) [above of=A2,right of=A2] {$X_{2}$};
 \node(B2) [below of=C2,right of=C2] {$T_{2}$};
 \node(D2) [above of=B2,right of=B2] {$Y_{2}$};
 \node(E2) [below of=D2,right of=D2] {$U_{2}$};
 \draw[->] (A1) to node [swap]{$$} (C1);
 \draw[->] (B1) to node [swap]{$$} (C1);
 \draw[->] (B1) to node [swap]{$$} (D1);
 \draw[->] (E1) to node [swap]{$$} (D1);
 \draw[->] (A2) to node [swap]{$$} (C2);
 \draw[->] (B2) to node [swap]{$$} (C2);
 \draw[->] (B2) to node [swap]{$$} (D2);
 \draw[->] (E2) to node [swap]{$$} (D2);
\end{tikzpicture}
\end{center}
in two different ways.  We call $\chi$ `globular' because its source and target 1-morphisms
are identities.   We need to check that $\chi$ indeed defines a 2-isomorphism in 
$\MMark$.

To do this, we need to show that
\begin{equation}
\label{eq:chi}
((H_1 \oplus H_2) \odot (G_1 \oplus G_2))\, \hat{\chi}_*  =  \hat{\chi}_* \, ((H_1 \odot G_1) \oplus (H_2 \odot G_2)) .
\end{equation}
To simplify notation, let $K =  (X_1 +_{T_1} Y_1) + (X_2 +_{T_2} Y_2)$ and 
$K'=(X_1 + X_2) +_{T_1 + T_2} (Y_1 + Y_2)$ so that $\hat{\chi} \colon K \stackrel{\sim}{\to} K'$.
Let 
\[   q \maps X_1 + X_2 + Y_1 + Y_2 \to K , \quad 
      q' \maps X_1 + X_2 + Y_1 + Y_2 \to K'  \]
be the canonical maps coming from the definitions of $K$ and $K'$ as colimits, and note that
\[  q' = \hat{\chi} q \]
by the universal property of the colimit.   A calculation using Eq.\ (\ref{eq:odot2}) implies that
\[    (H_1 \odot G_1) \oplus (H_2 \odot G_2) =
q_* \, ((H_1 \oplus H_2) \oplus (G_1 \oplus G_2)) \, q^* \]
and similarly 
\[ (H_1 \oplus H_2) \odot (G_1 \oplus G_2) 
= q'_* ((H_1 \oplus H_2) \oplus (G_1 \oplus G_2)) q'^*. \]
Together these facts give
\[  \begin{array}{ccl}  (H_1 \oplus H_2) \odot (G_1 \oplus G_2) 
&=& \hat{\chi}_* q_* \, ((H_1 \oplus H_2) \oplus (G_1 \oplus G_2)) \, q^* \hat{\chi}^* \\
&=&  \hat{\chi}_* \, ((H_1 \odot G_1) \oplus (H_2 \odot G_2)\, \hat{\chi}^*  .
\end{array}  \]
and since $\hat{\chi}$ is a bijection, $\hat{\chi}^*$ is the inverse of $\hat{\chi}_*$, 
so Eq.\ (\ref{eq:chi}) follows.

For the other globular 2-isomorphism, if $S$ and $T$ are finite sets, then $u(S+T)$ is given by
\[
\begin{tikzpicture}[scale=1.5]
\node (D) at (-5,-0.5) {$S+T$};
\node (E) at (-3,-0.5) {$(S+T,0_{S+T})$};
\node (F) at (-1,-0.5) {$S+T$};
\path[->,font=\scriptsize,>=angle 90]
(D) edge node[above] {$1_{S+T}$}(E)
(F) edge node[above] {$1_{S+T}$}(E);
\end{tikzpicture}
\]
while $u(S) \oplus u(T)$ is given by
\[
\begin{tikzpicture}[scale=1.5]
\node (D) at (-5,-0.5) {$S+T$};
\node (E) at (-3,-0.5) {$(S+T,0_S \oplus 0_T)$};
\node (F) at (-1,-0.5) {$S+T$};
\path[->,font=\scriptsize,>=angle 90]
(D) edge node[above] {$1_S + 1_T$}(E)
(F) edge node[above] {$1_S + 1_T$}(E);
\end{tikzpicture}
\]
so there is a globular 2-isomorphism $\mu$ between these, namely the identity 2-morphism.
All the commutative diagrams in the definition of symmetric monoidal double category \cite{Shulman} can be checked in a straightforward way.  
\end{proof}

\section{Black-boxing for open Markov processes}
\label{sec:black-boxing}

The general idea of `black-boxing' is to take a system and forget everything except the relation
between its inputs and outputs, as if we had placed it in a black box and were unable to see its inner workings.  Previous work of Pollard and the first author \cite{BP} constructed a black-boxing functor $\blacksquare \maps \Dynam \to \SemiAlgRel$ where $\Dynam$ is a category of finite sets and `open dynamical systems' and $\SemiAlgRel$ is a category of finite-dimensional real vector spaces and relations defined by polynomials and inequalities.   When we black-box such an open dynamical system, we obtain the relation between inputs and outputs that holds in steady state. 

A special case of an open dynamical system is an open Markov process as defined in this paper.  Thus, we could restrict the black-boxing functor $\blacksquare \maps \Dynam \to \SemiAlgRel$ to a category $\mathsf{Mark}$ with finite sets as objects and open Markov processes as morphisms.  Since the steady state behavior of a Markov process is \emph{linear}, we would get a functor $\blacksquare \maps \Mark \to \LinRel$ where $\LinRel$ is the category of finite-dimensional real vector spaces and \emph{linear} relations \cite{BE}.   However, we will go further and define black-boxing on the double category $\MMark$.   This will exhibit the relation between black-boxing and morphisms between open Markov processes.

To do this, we promote $\LinRel$ to a double category $\LLinRel$ with:
\begin{enumerate}
\item finite-dimensional real vector spaces $U,V,W,\dots$ as objects,
\item linear maps $f \maps V \to W$ as vertical 1-morphisms from $V$ to $W$,
\item linear relations $R \subseteq V \oplus W$  as horizontal 1-cells from $V$ to $W$,
\item squares 
\[
\begin{tikzpicture}[scale=1.5]
\node (D) at (-4,0.5) {$V_1$};
\node (E) at (-2,0.5) {$V_2$};
\node (F) at (-4,-1) {$W_1$};
\node (A) at (-2,-1) {$W_2$};
\node (B) at (-3,-0.25) {};
\path[->,font=\scriptsize,>=angle 90]
(D) edge node [above]{$R \subseteq V_1 \oplus V_2$}(E)
(E) edge node [right]{$g$}(A)
(D) edge node [left]{$f$}(F)
(F) edge node [above]{$S \subseteq W_1 \oplus W_2$} (A);
\end{tikzpicture}
\]
obeying $(f \oplus g)R \subseteq S$ as 2-morphisms. 
\end{enumerate}
The last item deserves some explanation.  A preorder is a category such that for any 
pair of objects $x,y$ there exists at most one morphism $\alpha \maps x \to y$.  When such
a morphism exists we usually write $x \le y$.  Similarly there is a kind of double category for which  given any `frame'
\[
\begin{tikzpicture}[scale=1]
\node (D) at (-4,0.5) {$A$};
\node (E) at (-2,0.5) {$B$};
\node (F) at (-4,-1) {$C$};
\node (A) at (-2,-1) {$D$};
\node (B) at (-3,-0.25) {};
\path[->,font=\scriptsize,>=angle 90]
(D) edge node [above]{$M$}(E)
(E) edge node [right]{$g$}(A)
(D) edge node [left]{$f$}(F)
(F) edge node [above]{$N$} (A);
\end{tikzpicture}
\]
there exists at most one 2-morphism 
\[
\begin{tikzpicture}[scale=1]
\node (D) at (-4,0.5) {$A$};
\node (E) at (-2,0.5) {$B$};
\node (F) at (-4,-1) {$C$};
\node (A) at (-2,-1) {$D$};
\node (B) at (-3,-0.25) {$\Downarrow \alpha$};
\path[->,font=\scriptsize,>=angle 90]
(D) edge node [above]{$M$}(E)
(E) edge node [right]{$g$}(A)
(D) edge node [left]{$f$}(F)
(F) edge node [above]{$N$} (A);
\end{tikzpicture}
\]
filling this frame. For lack of a better term let us call this a \define{degenerate} double category.  
Item (iv) implies that $\LLinRel$ will be degenerate in this sense.

In $\LLinRel$, composition of vertical 1-morphisms is the usual composition of linear maps, while  
composition of horizontal 1-cells is the usual composition of linear relations.  Since composition of
linear relations obeys the associative and unit laws strictly, $\LLinRel$ will be a \emph{strict} double category.   Since $\LLinRel$ is degenerate, there is at most one way to define the vertical composite of 2-morphisms 
\[
\begin{tikzpicture}[scale=1.5]
\node (D) at (-4,0) {$U_1$};
\node (E) at (-2,0) {$U_2$};
\node (F) at (-4,-1.5) {$V_1$};
\node (A) at (-2,-1.5) {$V_2$};
\node (B) at (-3,-0.75) {$\Downarrow \alpha$};
\node (C) at (-4,-3) {$W_1$};
\node (G) at (-2,-3) {$W_2$};
\node (H) at (-3,-2.25) {$\Downarrow \beta$};
\node (I) at (-1,-1.5) {$=$};
\node (J) at (0,-0.5) {$U_1$};
\node (K) at (2,-0.5) {$U_2$};
\node (L) at (0,-2.5) {$W_1$};
\node (M) at (2,-2.5) {$W_2$};
\node (O) at (1,-1.5) {$\Downarrow \beta \alpha$};
\path[->,font=\scriptsize,>=angle 90]
(D) edge node [above]{$R \subseteq U_1 \oplus U_2$}(E)
(E) edge node [right]{$g$}(A)
(D) edge node [left]{$f$}(F)
(F) edge node [left] {$f'$}(C)
(C) edge node [above] {$T \subseteq W_1 \oplus W_2$} (G)
(A) edge node [right] {$g'$} (G)
(F) edge node [above]{$S \subseteq V_1 \oplus V_2$} (A)
(J) edge node [above] {$R \subseteq U_1 \oplus U_2$} (K)
(K) edge node [right] {$g' g$} (M)
(J) edge node [left] {$f' f$} (L)
(L) edge node [above] {$T \subseteq W_1 \oplus W_2$} (M);
\end{tikzpicture}
\]
so we need merely check that a 2-morphism $\beta\alpha$ filling the frame at right exists.  This amounts to noting that
\[     (f \oplus g)R \subseteq S, \; (f' \oplus g')S \subseteq T \; \implies \;
(f' \oplus g')(f \oplus g)R \subseteq T .\]
Similarly, there is at most one way to define the horizontal composite of 2-morphisms
\[
\begin{tikzpicture}[scale=1.5]
\node (D) at (-3.5,0) {$V_1$};
\node (E) at (-2,0) {$V_2$};
\node (F) at (-3.5,-1.5) {$W_1$};
\node (A) at (-2,-1.5) {$W_2$};
\node (B) at (-2.75,-0.75) {$\Downarrow \alpha$};
\node (J) at (-0.5,0) {$V_3$};
\node (L) at (-0.5,-1.5) {$W_3$};
\node (O) at (-1.25,-0.75) {$\Downarrow \alpha'$};
\node (I) at (0.25,-0.75) {$=$};
\node (D') at (1,0) {$V_1$};
\node (E') at (2.5,0) {$V_3$};
\node (F') at (1,-1.5) {$W_1$};
\node (A') at (2.5,-1.5) {$W_3$};
\node (B') at (1.75,-0.75) {$\Downarrow \alpha' \circ \alpha$};
\path[->,font=\scriptsize,>=angle 90]
(D) edge node [above]{$R \subseteq V_1 \oplus V_2$}(E)
(E) edge node [right]{$g$}(A)
(D) edge node [left]{$f$}(F)
(F) edge node [above]{$S \subseteq W_1 \oplus W_2$} (A)
(E) edge node [above]{$R' \subseteq V_2 \oplus V_3$} (J)
(J) edge node [right]{$h$} (L)
(A) edge node [above]{$S' \subseteq W_2 \oplus W_3$} (L)
(D') edge node [above]{$R'R \subseteq V_1 \oplus V_3$}(E')
(D') edge node [left]{$f$}(F')
(F') edge node [above]{$S'S \subseteq W_1 \oplus W_3$} (A')
(E') edge node [right]{$h$} (A');
\end{tikzpicture}
\]
so we need merely check that a filler $\alpha' \circ \alpha$ exists, which amounts to noting that
\[   (f \oplus g)R \subseteq S , \; (g \oplus h)R' \subseteq S' \; \implies \;
(f \oplus h)(R'R) \subseteq S'S . \]

\begin{thm}
\label{thm:LLinRel}
There exists a strict double category $\LLinRel$ with the above properties.
\end{thm}

\begin{proof}
The category of objects $\LLinRel_0$ has finite-dimensional real vector spaces as objects and linear maps as morphisms.  The category of arrows $\LLinRel_1$ has linear relations as objects 
and squares 
\[
\begin{tikzpicture}[scale=1.5]
\node (D) at (-4,0.5) {$V_1$};
\node (E) at (-2,0.5) {$V_2$};
\node (F) at (-4,-1) {$W_1$};
\node (A) at (-2,-1) {$W_2$};
\node (B) at (-3,-0.25) {};
\path[->,font=\scriptsize,>=angle 90]
(D) edge node [above]{$R \subseteq V_1 \oplus V_2$}(E)
(E) edge node [right]{$g$}(A)
(D) edge node [left]{$f$}(F)
(F) edge node [above]{$S \subseteq W_1 \oplus W_2$} (A);
\end{tikzpicture}
\]
with $(f \oplus g)R \subseteq S$ as morphisms.  The source and target functors $s,t \maps \LLinRel_1 \to \LLinRel_0$ are clear. The identity-assigning functor $u  \maps \LLinRel_0 \to \LLinRel_1$ sends a finite-dimensional real vector space $V$ to the identity map $1_V$ and a linear map $f \maps V \to W$ to the unique 2-morphism
\[
\begin{tikzpicture}[scale=1.5]
\node (D) at (-3.75,0.5) {$V$};
\node (E) at (-2.25,0.5) {$V$};
\node (F) at (-3.75,-1) {$W$};
\node (A) at (-2.25,-1) {$W$.};
\node (B') at (-3,-0.25) {};
\path[->,font=\scriptsize,>=angle 90]
(D) edge node [above]{$1_V$}(E)
(E) edge node [right]{$f$}(A)
(D) edge node [left]{$f$}(F)
(F) edge node [above]{$1_W$} (A);
\end{tikzpicture}
\]
The composition functor $\odot \maps \LLinRel_1 \times_{\LLinRel_0} \LLinRel_1 \to \LLinRel_1$
acts on objects by the usual composition of linear relations, and it acts on 2-morphisms by horizontal composition as described above.  These functors can be shown to obey all the axioms of a double category.   In particular, because $\LLinRel$ is degenerate, all the required equations between 2-morphisms, such as the interchange law, hold automatically.
\end{proof}

Next we make $\LLinRel$ into a symmetric monoidal double category.  To do this, we first give $\LLinRel_0$ the structure of a symmetric monoidal category.  We do this using a specific choice of direct sum for each pair of finite-dimensional real vector spaces as the tensor product, and a specific 0-dimensional vector space as the unit object.   Then we give $\LLinRel_1$ a symmetric monoidal structure as follows.  Given linear relations $R_1 \subseteq V_1 \oplus W_1$ and $R_2 \subseteq V_2 \oplus W_2$, we define their direct sum by
\[  R_1 \oplus R_2 = \{(v_1,v_2,w_1,w_2) : \; (v_1,w_1) \in R_1, (v_2,w_2) \in R_2 \} \subseteq V_1 \oplus V_2 \oplus W_1 \oplus W_2. \]
Given two 2-morphisms in $\LLinRel_1$:
\[
\begin{tikzpicture}[scale=1.5]
\node (D) at (-4,0.5) {$V_1$};
\node (E) at (-2,0.5) {$V_2$};
\node (F) at (-4,-1) {$W_1$};
\node (A) at (-2,-1) {$W_2$};
\node (D') at (-0.5,0.5) {$V'_1$};
\node (E') at (1.5,0.5) {$V'_2$};
\node (F') at (-0.5,-1) {$W'_1$};
\node (A') at (1.5,-1) {$W'_2$};
\node (B') at (0.5,-0.25) {$\Downarrow \alpha'$};
\node (B) at (-3,-0.25) {$\Downarrow \alpha$};
\path[->,font=\scriptsize,>=angle 90]
(D) edge node [above]{$R \subseteq V_1 \oplus V_2$}(E)
(E) edge node [right]{$g$}(A)
(D) edge node [left]{$f$}(F)
(F) edge node [above]{$S \subseteq W_1 \oplus W_2$} (A)
(D') edge node [above]{$R' \subseteq V'_1 \oplus V'_2$}(E')
(E') edge node [right]{$g'$}(A')
(D') edge node [left]{$f'$}(F')
(F') edge node [above]{$S' \subseteq W'_1 \oplus W'_2$} (A');
\end{tikzpicture}
\]
there is at most one way to define their direct sum
\[
\begin{tikzpicture}[scale=1.5]
\node (D) at (-3.5,0.5) {$V_1 \oplus V'_1$};
\node (E) at (-.5,0.5) {$V_2 \oplus V'_2$};
\node (F) at (-3.5,-1) {$W_1 \oplus W'_1$};
\node (A) at (-.5,-1) {$W_2 \oplus W'_2$};
\node (B') at (-2,-0.25) {$\Downarrow \alpha \oplus \alpha'$};
\path[->,font=\scriptsize,>=angle 90]
(D) edge node [above]{$R \oplus R' \subseteq V_1 \oplus V'_1 \oplus V_2 \oplus V'_2$}(E)
(E) edge node [right]{$g \oplus g'$}(A)
(D) edge node [left]{$f \oplus f'$}(F)
(F) edge node [above]{$S \oplus S' \subseteq W_1 \oplus W'_1 \oplus W_2 \oplus W'_2$} (A);
\end{tikzpicture}
\]
because $\LLinRel$ is degenerate.   To show that $\alpha \oplus \alpha'$ exists, we need
merely note that
\[    (f \oplus g) R \subseteq S, \; (f' \oplus g')R' \subseteq S' \; \implies \; 
(f \oplus f' \oplus g \oplus g') (R \oplus R') \subseteq S \oplus S'.  \]

\begin{thm}
\label{thm:LLinRel2}
The double category $\LLinRel$ can be given the structure of a symmetric monoidal double category with the above properties.
\end{thm}

\begin{proof}
We have described $\LLinRel_0$ and $\LLinRel_1$ as symmetric monoidal categories.
The source and target functors $s,t \maps \LLinRel_1 \to \LLinRel_0$ are strict symmetric
monoidal functors.  The required globular 2-isomorphisms $\chi$ and $\mu$ are defined as
follows.   Given four horizontal 1-cells 
\[   R_1 \subseteq U_1 \oplus V_1, \quad R_2 \subseteq V_1 \oplus W_1, \]
\[ S_1 \subseteq U_2 \oplus V_2, \quad S_2 \subseteq V_2 \oplus W_2, \]
the globular 2-isomorphism $\chi \maps (R_2 \oplus S_2)(R_1 \oplus S_1) \Rightarrow (R_2 R_1) \oplus (S_2 S_1)$ is the identity 2-morphism
\[
\begin{tikzpicture}[scale=1.5]
\node (D) at (-4,0.5) {$U_1 \oplus U_2$};
\node (E) at (-1,0.5) {$W_1 \oplus W_2$};
\node (F) at (-4,-1) {$U_1 \oplus U_2$};
\node (A) at (-1,-1) {$W_1 \oplus W_2$.};
\node (B') at (-2.5,-0.25) {$$};
\path[->,font=\scriptsize,>=angle 90]
(D) edge node [above]{$(R_2 \oplus S_2)(R_1 \oplus S_1)$}(E)
(E) edge node [right]{$1$}(A)
(D) edge node [left]{$1$}(F)
(F) edge node [above]{$(R_2 R_1) \oplus (S_2 S_1)$} (A);
\end{tikzpicture}
\]
The globular 2-isomorphism $\mu \maps u(V \oplus W) \Rightarrow u(V) \oplus u(W)$ is
the identity 2-morphism
\[
\begin{tikzpicture}[scale=1.5]
\node (D) at (-4,0.5) {$V \oplus W$};
\node (E) at (-2,0.5) {$V \oplus W$};
\node (F) at (-4,-1) {$V \oplus W$};
\node (A) at (-2,-1) {$V \oplus W$.};
\node (B') at (-3,-0.25) {$$};
\path[->,font=\scriptsize,>=angle 90]
(D) edge node [above]{$1_{V \oplus W} $}(E)
(E) edge node [right]{$1$}(A)
(D) edge node [left]{$1$}(F)
(F) edge node [above]{$1_V \oplus 1_W$} (A);
\end{tikzpicture}
\]
All the commutative diagrams in the definition of symmetric monoidal double category \cite{Shulman} can be checked straightforwardly.   In particular, all diagrams of
2-morphisms commute automatically because $\LLinRel$ is degenerate.  \end{proof}

Theorems \ref{thm:LLinRel} and \ref{thm:LLinRel2} could be proved more generally,
replacing linear relations with relations in an arbitrary regular category.  However, here we need these results only to set the stage for defining the symmetric monoidal double functor $\blacksquare \maps \MMark \to \LLinRel$.  We proceed as follows:
\begin{enumerate}
\item On objects: for a finite set $S$, we define $\blacksquare(S)$ to be the vector space $\R^S \oplus \R^S$.
\item On horizontal 1-cells: for an open Markov process $S \stackrel{i}{\rightarrow} (X,H) \stackrel{o}{\leftarrow} T$, we define its black-boxing as in Def.\ \ref{defn:black-boxing}:
\[ \blacksquare(S \stackrel{i}{\rightarrow} (X,H) \stackrel{o}{\leftarrow} T) = \]
\[ \{ (i^*(v),I,o^*(v),O) : \; v \in \R^X,  I \in \R^S, O \in \R^T \, 
\textrm{ and } \, H(v) + i_*(I) - o_*(O) = 0\}. \]
\item On vertical 1-morphisms: for a map $f \maps S\to S'$, we define $\blacksquare(f) \maps \R^S \oplus \R^S \to \R^{S'} \oplus \R^{S'}$ to be the linear map $f_* \oplus f_*$.
\end{enumerate}

What remains to be done is define how $\blacksquare$ acts on 2-morphisms of $\MMark$.   This describes the relation between steady state input and output concentrations and flows of a coarse-grained open Markov process in terms of the corresponding relation for the original process:

\begin{lem}
\label{lem:black-boxing_2-morphisms}
Given a 2-morphism 
\[
\begin{tikzpicture}[scale=1.5]
\node (D) at (-4,0.5) {$S$};
\node (E) at (-3,0.5) {$(X,H)$};
\node (F) at (-2,0.5) {$T$};
\node (G) at (-3,-1) {$(X',H')$};
\node (A) at (-4,-1) {$S'$};
\node (B) at (-2,-1) {$T'$,};
\path[->,font=\scriptsize,>=angle 90]
(D) edge node [left]{$f$}(A)
(F) edge node [right]{$g$}(B)
(D) edge node [above] {$i$}(E)
(F) edge node [above] {$o$}(E)
(A) edge node[left] [above] {$i'$} (G)
(B) edge node[left] [above] {$o'$} (G)
(E) edge node[left] {$p$}(G);
\end{tikzpicture}
\]
in $\MMark$, there exists a (unique) 2-morphism
\[
\begin{tikzpicture}[scale=1.5]
\node (D) at (-4.5,0.5) {$\blacksquare(S)$};
\node (E) at (-1.5,0.5) {$\blacksquare(T)$};
\node (F) at (-4.5,-1) {$\blacksquare(S')$};
\node (A) at (-1.5,-1) {$\blacksquare(T')$};
\path[->,font=\scriptsize,>=angle 90]
(D) edge node [above]{$\blacksquare(S \stackrel{i}{\rightarrow} (X,H) \stackrel{o}{\leftarrow} T)$}(E)
(E) edge node [right]{$\blacksquare(g)$}(A)
(D) edge node [left]{$\blacksquare(f)$}(F)
(F) edge node [above]{$ \blacksquare(S' \stackrel{i'}{\rightarrow} (X',H') \stackrel{o'}{\leftarrow} T')$} (A);
\end{tikzpicture}
\]
\vskip -0.5em \noindent 
in $\LLinRel$.
\end{lem}

\begin{proof}  
Since $\LLinRel$ is degenerate, if there exists a 2-morphism of the claimed kind it is
automatically unique.  To prove that such a 2-morphism exists, it suffices to prove
\[    (i^*(v),I,o^*(v),O) \in V \; \implies \;  (f_* i^*(v),f_*(I), g_* o^*(v),g_*(O)) \in W \]
where
\[  V = \blacksquare(S \stackrel{i}{\rightarrow} (X,H) \stackrel{o}{\leftarrow} T) = \]
\[  \{ (i^*(v),I,o^*(v),O) : \; v \in \R^X,  I \in \R^S, O \in \R^T \, 
\textrm{ and } \, H(v) + i_*(I) - o_*(O) = 0\} \]
and
\[ W = \blacksquare(S' \stackrel{i'}{\rightarrow} (X',H') \stackrel{o'}{\leftarrow} T') = \]
\[ \{ (i'^*(v'),I',o'^*(v'),O') : \; v' \in \R^{X'},  I' \in \R^{S'}, O' \in \R^{T'} \, 
\textrm{ and } \, H'(v') + i'_*(I') - o'_*(O') = 0\} .\]

To do this, assume $  (i^*(v),I,o^*(v),O) \in V$, which implies that
\begin{equation}
  H(v) + i_*(I) - o_*(O) = 0 .    \label{eq:master_1}
\end{equation}
Since the commuting squares in $\alpha$ are pullbacks,  Lemma \ref{lem:beck-chevalley} implies that
\[
     f_* i^* = i'^* p_*    , \qquad
     g_* o^* = o'^* p_*  . 
\]
Thus
\[     (f_* i^*(v),f_*(I), g_* o^*(v),g_*(O)) = (i'^* p_*(v), f_*(I), o'^* p_*(v),g_*(O) )  \]
and this is an element of $W$ as desired if
\begin{equation}   
    H' p_*(v) + i'_* f_*(I) - o'_* g_*(O) = 0 .   \label{eq:master_2} 
\end{equation}
To prove Eq.\ \eqref{eq:master_2}, note that
\[
\begin{array}{ccl}
   H' p_*(v) + i'_* f_*(I) - o'_* g_*(O) &=& p_*H (v) + p_* i_*(I) - p_* o_* (O)  \\ 
                                                          &=& p_*(H(v) + i_*(I) - o_*(O)) 
\end{array}
\]
where in the first step we use the fact that the squares in $\alpha$ commute, together with the fact that $H' p_* = p_* H$.   Thus, Eq.\ \eqref{eq:master_1} implies Eq.\ \eqref{eq:master_2}. 
\end{proof}

The following result is a special case of a result by Pollard and the first author on black-boxing open dynamical systems \cite{BP}.   To make this paper self-contained we adapt the proof
to the case at hand:

\begin{lem}
\label{lem:black-boxing_symmoncat}
The black-boxing of a composite of two open Markov processes equals the composite of
their black-boxings.
\end{lem}

\begin{proof}
 Consider composable open Markov processes  
\[         S \stackrel{i}\longrightarrow (X,H) \stackrel{o}\longleftarrow T, \qquad
           T \stackrel{i'}\longrightarrow (Y,G) \stackrel{o'}\longleftarrow U .\]
To compose these, we first form the pushout
\[
    \xymatrix{
      && X +_T Y \\
      & X \ar[ur]^{j} && Y \ar[ul]_{k} \\
      \quad S\quad \ar[ur]^{i} && T \ar[ul]_{o} \ar[ur]^{i'} &&\quad U \quad \ar[ul]_{o'}
    }
\]
Then their composite is
\[ S \stackrel{j i}{\longrightarrow} (X +_T Y , H \odot G) \stackrel{k o'}{\longleftarrow} U \]
where 
\[    H \odot G = j_* H j^* + k_* G  k^*  .\]

To prove that $\blacksquare$ preserves composition, we first show that 
\[ \blacksquare(Y,G) \; \blacksquare(X,H) \subseteq \blacksquare(X+_T Y,H \odot G) .\]  
Thus, given
\[     (i^*(v),I,o^*(v),O) \in \blacksquare(X,H), \qquad  ({i'}^*(v'),I',{o'}^*(v'),O') \in \blacksquare(Y,G) \]
with 
\[   o^*(v) = {i'}^*(v'), \qquad O = I' \]
we need to prove that 
\[     (i^*(v),I,{o'}^*(v'),O') \in \blacksquare(X +_T Y, H \odot G). \]
To do this, it suffices to find $w \in \R^{X +_T Y}$ such that 
\[   (i^*(v),I,{o'}^*(v'),O') = ((ji)^*(w), I, {(ko')}^*(w), O') \]
and $w$ is a steady state of $(X+_T Y,H \odot G)$ with inflows $I$ and outflows $O'$.

Since  $o^*(v) = {i'}^*(v'),$ this diagram commutes:
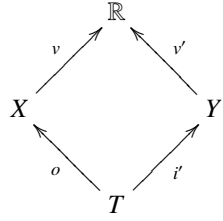
\[
    \xymatrix{
      && \R \\
      & X \ar[ur]^{v} && Y \ar[ul]_{v'} \\
       && T \ar[ul]^{o} \ar[ur]_{i'} &&
    }
\]
so by the universal property of the pushout there is a unique map $w \maps X +_T Y \to \R$ such that this commutes:
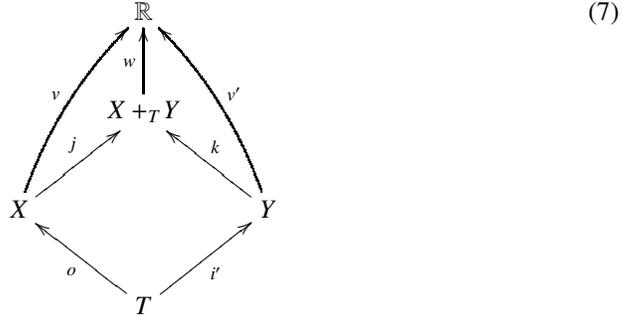
\begin{equation}
\label{eq:pushout}
    \xymatrix{
      && \R \\
     && X+_T\!Y \ar[u]^w \\      
      & X \ar@/^/[uur]^{v} \ar[ur]^{j} && Y \ar@/_/[uul]_{v'} \ar[ul]_{k} \\
       && T \ar[ul]^{o} \ar[ur]_{i'} &&
    }
\end{equation}
This simply says that because the functions $v$ and $v'$ agree on the `overlap' of our two
open Markov processes, we can find a function $w$ that restricts
to $v$ on $X$ and $v'$ on $Y$.

We now prove that $w$ is a steady state of the composite open Markov process with
inflows $I$ and outflows $O'$:
\begin{equation}
   \label{eq:steady_state_3}
   (H \odot G)(w) + (ji)_*(I) - (ko')_*(O') = 0.
\end{equation}
To do this we use the fact that $v$ is a steady state of $S \stackrel{i}\rightarrow (X,H) \stackrel{o}\leftarrow T$ with inflows $I$ and outflows $O$:
\begin{equation}\label{eq:steady_state_1}
   H(v) + i_*(I) - o_*(O) = 0
\end{equation}
and $v'$ is a steady state of $T \stackrel{i'}\rightarrow (Y,G) \stackrel{o'}\leftarrow U$ with inflows $I'$ and outflows $O'$:
\begin{equation}\label{eq:steady_state_2}
   G(v') + i'_*(I') - o'_*(O') = 0.
\end{equation}
We push Eq.\ \eqref{eq:steady_state_1} forward along $j$, push Eq.\
\eqref{eq:steady_state_2} forward along $k$, and sum them:
\[   j_*(H(v))  + (ji)_*(I) - (jo)_*(O) + k_*(G(v')) + (ki')_*(I') - (ko')_*(O') = 0. \]
Since $O = I'$ and $jo = ki'$, two terms cancel, leaving us with
\[     j_*(H(v))  + (ji)_*(I) + k_*(G(v')) - (ko')_*(O') = 0. \]
Next we combine the terms involving the infinitesimal stochastic operators $H$ and $G$, with the help of Eq.\ \eqref{eq:pushout} and the definition of $H \odot G$:
\begin{equation}
\label{eq:u}
   \begin{array}{ccl}
  j_*(H(v)) + k_*(G(v')) &=& (j_* H j^* + k_* G k^*)(w) \\
                                    &=& (H \odot G)(w)  .
\end{array}
\end{equation}
This leaves us with
\[         (H \odot G)(w) +  (ji)_*(I) - (ko')_*(O') = 0 \]
which is Eq.\ \eqref{eq:steady_state_3}, precisely what we needed to show.

To finish showing that $\blacksquare$ is a functor, we need to show that 
\[   \blacksquare(X+_T Y,H \odot G) \subseteq \blacksquare(Y,G) \; \blacksquare(X,H)  .\] 
So, suppose we have 
\[    ((ji)^*(w), I, {(ko')}^*(w), O') \in \blacksquare(X+_T Y,H \odot G) .\]
We need to show
\begin{equation}
\label{eq:composite}
  ((ji)^*(w), I, {(ko')}^*(w), O') = (i^*(v),I,o'^*(v'),O')
\end{equation}
where 
\[     (i^*(v),I,o^*(v),O) \in \blacksquare(X,H), \qquad  (i'^*(v'),I',o'^*(v'),O') \in \blacksquare(Y,G) \]
and
\[   o^*(v) = i'^*(v'), \qquad O = I' .\]

To do this, we begin by choosing
\[   v = j^*(w), \qquad v' = k^*(w) .\]
This ensures that Eq.\ \eqref{eq:composite} holds, and since $jo = ki'$, it also ensures that 
\[  o^*(v) = (jo)^*(w) = (ki')^*(w) = {i'}^*(v')  .\]
To finish the job, we need to find an element $O = I' \in \R^T$ such that $v$ is a steady state of $(X,H)$ with inflows $I$ and outflows $O$ and $v'$ is a steady state of $(Y,G)$ with inflows $I'$ and outflows $O'$.  Of course, we are given the fact that $w$ is a steady state of $(X+_T Y,H \odot G)$ with inflows $I$ and outflows $O'$.   

In short, we are given Eq.\ \eqref{eq:steady_state_3}, and we seek $O = I'$ such that Eqs.\ \eqref{eq:steady_state_1} and \eqref{eq:steady_state_2} hold.  Thanks to our choices of $v$ and $v'$,  we can use Eq.\ \eqref{eq:u} and rewrite Eq.\ \eqref{eq:steady_state_3} as
\begin{equation}
\label{eq:steady_state_3'}
  j_*(H(v) + i_*(I)) \; + \; k_*(G(v') - o'_*(O')) = 0 .  
\end{equation}
Eqs.\ \eqref{eq:steady_state_1} and \eqref{eq:steady_state_2} say that
\begin{equation}
\label{eq:steady_state_1'2'}
\begin{array}{lcl}
   H(v) + i_*(I) - o_*(O) &=& 0 \\ 
   G(v') + i'_*(I') - o'_*(O') &=& 0.
\end{array}
\end{equation}

Now we use the fact that 
\[
    \xymatrix{
      & X +_T Y \\
       X \ar[ur]^{j} && Y \ar[ul]_{k} \\
       & T \ar[ul]^{o} \ar[ur]_{i'} &
    }
\]
is a pushout.  Applying the `free vector space on a finite set' functor, which preserves colimits, this implies that
\[
    \xymatrix{
      & \R^{X +_T Y} \\
       \R^X \ar[ur]^{j_*} && \R^{Y} \ar[ul]_{k_*} \\
       & \R^T \ar[ul]^{o_*} \ar[ur]_{i'_*} &
    }
\]
is a pushout in the category of vector spaces.   Since a pushout is formed by taking first a coproduct and then a coequalizer, this implies that 
\[
     \xymatrix{
      \R^T \ar@<-.5ex>[rr]_-{(0,i'_*)} \ar@<.5ex>[rr]^-{(o_*,0)} && \R^X \oplus \R^{Y} \ar[rr]^{j_* + k_*}
   && \R^{X +_T Y}
}
\]
is a coequalizer.  Thus, the kernel of $j_* + k_*$ is the image of $(o_*,0) - (0,i'_*)$.   Eq.\ \eqref{eq:steady_state_3'} says precisely that 
\[    (H(v) + i_*(I), G(v') - o'_*(O')) \in \ker(j_* + k_*)  .\]
Thus, it is in the image of $o_* - i'_*$.  In other words, there exists some element $O = I' \in \R^T$
such that 
\[   (H(v) + i_*(I), G(v') - o'_*(O')) = (o_*(O), -i'_*(I')).\]
This says that Eqs.\ \eqref{eq:steady_state_1} and \eqref{eq:steady_state_2} hold, as desired.
\end{proof}

This is the main result of this paper:

\begin{thm}
\label{thm:main}
There exists a symmetric monoidal double functor $\blacksquare \maps \MMark \to \LLinRel$ with the following behavior:

\begin{enumerate}
\item Objects: $\blacksquare$ sends any finite set $S$ to the vector space $\R^{S} \oplus \R^{S}$.
\item Vertical 1-morphisms: $\blacksquare$ sends any map $f \maps S \to S'$ to the linear
map \hfill \break  $f_* \oplus f_* \maps \R^{S} \oplus \R^S \to \R^{S'} \oplus \R^{S'}$.
\item Horizontal 1-cells: $\blacksquare$ sends any open Markov process $S \stackrel{i}{\rightarrow} (X,H) \stackrel{o}{\leftarrow} T$ to the linear relation given in Def.\ \ref{defn:black-boxing}:
\[ \blacksquare(S \stackrel{i}{\rightarrow} (X,H) \stackrel{o}{\leftarrow} T) = \]
\[ \{ (i^*(v),I,o^*(v),O) : \; H(v) + i_*(I) - o_*(O) = 0 \textrm{ for some } I \in \R^S, v \in \R^X, O \in \R^T \}. \]
\item 2-Morphisms: $\blacksquare$ sends any morphism of open Markov processes
\[
\begin{tikzpicture}[scale=1.5]
\node (D) at (-4,0.5) {$S$};
\node (E) at (-3,0.5) {$(X,H)$};
\node (F) at (-2,0.5) {$T$};
\node (G) at (-3,-1) {$(X',H')$};
\node (A) at (-4,-1) {$S'$};
\node (B) at (-2,-1) {$T'$};
\path[->,font=\scriptsize,>=angle 90]
(D) edge node [left]{$f$}(A)
(F) edge node [right]{$g$}(B)
(D) edge node [above] {$i$}(E)
(F) edge node [above] {$o$}(E)
(A) edge node[above] {$i'$} (G)
(B) edge node[above] {$o'$} (G)
(E) edge node[left] {$p$}(G);
\end{tikzpicture}
\]
to the 2-morphism in $\LLinRel$ given in Lemma \ref{lem:black-boxing_2-morphisms}:
\[
\begin{tikzpicture}[scale=1.5]
\node (D) at (-4.5,0.5) {$\blacksquare(S)$};
\node (E) at (-1.5,0.5) {$\blacksquare(T)$};
\node (F) at (-4.5,-1) {$\blacksquare(S')$};
\node (A) at (-1.5,-1) {$\blacksquare(T')$.};
\path[->,font=\scriptsize,>=angle 90]
(D) edge node [above]{$\blacksquare(S \stackrel{i}{\rightarrow} (X,H) \stackrel{o}{\leftarrow} T)$}(E)
(E) edge node [right]{$\blacksquare(g)$}(A)
(D) edge node [left]{$\blacksquare(f)$}(F)
(F) edge node [above]{$ \blacksquare(S' \stackrel{i'}{\rightarrow} (X',H') \stackrel{o'}{\leftarrow} T')$} (A);
\end{tikzpicture}
\]
\end{enumerate}
\end{thm}

\begin{proof}
First we must define functors $\blacksquare_{0} \maps \MMark_0 \to \LLinRel_0$ and $\blacksquare_1 \maps \MMark_{1} \to \LLinRel_{1}$. The functor $\blacksquare_0$ is defined on finite sets and maps between these as described in (i) and (ii) of the theorem statement, while $\blacksquare_1$ is defined on open Markov processes and morphisms between these as described in (iii) and (iv).  Lemma \ref{lem:black-boxing_2-morphisms} shows that $\blacksquare_1$ is well-defined on morphisms between open Markov processes; given this is it easy to check that $\blacksquare_1$ is a functor.  One can verify that $\blacksquare_0$ and $\blacksquare_1$ combine to define a double functor $\blacksquare \maps \MMark \to \LLinRel$: the hard part is checking that horizontal composition of open Markov processes is preserved, but this was shown in Lemma \ref{lem:black-boxing_symmoncat}.  Horizontal composition of 2-morphisms is automatically preserved because $\LLinRel$ is degenerate.

To make $\blacksquare$ into a symmetric monoidal double functor we need to make
$\blacksquare_0$ and $\blacksquare_1$ into symmetric monoidal functors, which we do using these extra structures:
\begin{itemize}
\item an isomorphism in $\LLinRel_0$ between $\{0\}$ and $\blacksquare(\emptyset)$,
\item a natural isomorphism between $\blacksquare(S)  \oplus \blacksquare(S')$ and 
 $\blacksquare(S + S')$ for any two objects $S,S' \in \MMark_0$,
\item an isomorphism in $\LLinRel_1$ between the unique linear relation $\{0\} \to \{0\}$ and
$\blacksquare(\emptyset \to (\emptyset, 0_\emptyset) \leftarrow \emptyset)$, and
\item a natural isomorphism between 
\[
\blacksquare((S \to (X,H) \leftarrow T) \; \oplus \; \blacksquare(S' \to (X',H') \leftarrow T')  \]
and
\[  \blacksquare(S + S' \to (X+X',H \oplus H') \leftarrow T + T') \]
for any two objects $S \to (X,H) \leftarrow T$, $S' \to (X',H') \leftarrow T'$ of $\MMark_1$.
\end{itemize}
There is an evident choice for each of these extra structures, and it is straightforward to check
that they not only  make $\blacksquare_0$ and $\blacksquare_1$ into symmetric monoidal functors but also meet the extra requirements for a symmetric monoidal double functor listed in Shulman's
paper \cite{Shulman}.  In particular, all diagrams of 2-morphisms commute automatically because $\LLinRel$ is degenerate.
\end{proof}

\section{A bicategory of open Markov processes}
\label{sec:bicat}

In Thm.\ \ref{thm:MMark_symmetric_monoidal}, we constructed a symmetric monoidal double category $\MMark$ with
\begin{enumerate}
\item finite sets as objects,
\item maps between finite sets as vertical 1-morphisms, 
\item open Markov processes as horizontal 1-cells, and
\item morphisms of open Markov processes as 2-morphisms.
\end{enumerate}
Using the following result of Shulman \cite{Shulman}, we can obtain a symmetric monoidal bicategory $\bold{Mark}$ with
\begin{enumerate}
\item finite sets as objects,
\item open Markov processes as morphisms,
\item morphisms of open Markov processes as 2-morphisms.
\end{enumerate}
To do this, we need to check that the symmetric monoidal double category $\MMark$ is `isofibrant'---a concept we explain in the proof of Lemma \ref{lem:isofibrant}.  The bicategory $\bold{Mark}$ then arises as the `horizontal bicategory' of the double category $\MMark$.  

\begin{defn}
Let $\lD$ be a double category. Then the $\textbf{horizontal bicategory}$ of $\lD$, which we denote as $H(\lD)$, is the bicategory with
\begin{enumerate}
\item objects of $\lD$ as objects,
\item  horizontal 1-cells of $\lD$ as 1-morphisms, 
\item globular 2-morphisms of $\lD$ (i.e., 2-morphisms with identities as their source and target)  as 2-morphisms,
\end{enumerate}
and vertical and horizontal composition, identities, associators and unitors arising from those in 
$\lD$.
\end{defn}

\begin{thm}[\textbf{Shulman}] \label{Shulman}
Let $\lD$ be an isofibrant symmetric monoidal double category. Then $H(\bold{\lD})$ is a symmetric monoidal bicategory, where $H(\bold{\lD})$ is the horizontal bicategory of $\lD$.
\end{thm}

\begin{lem}
\label{lem:isofibrant}
The symmetric monoidal double category $\MMark$ is isofibrant.
\end{lem}

\begin{proof}
In what follows, all unlabeled arrows are identities. To show that $\MMark$ is isofibrant, we need to show that every vertical 1-isomorphism has both a companion and a conjoint \cite{Shulman}. Given a vertical 1-isomorphism $f \maps S \to S'$, meaning a bijection between finite sets, then a companion of $f$ is given by the horizontal 1-cell:
\[
\begin{tikzpicture}[scale=1.5];
\node (D) at (-4.5,.5) {$S$};
\node (E) at (-3,.5) {$(S',0_{S'})$};
\node (F) at (-1.5,.5) {$S'$};
\path[->,font=\scriptsize,>=angle 90]
(D) edge node [above] {$f$}(E)
(F) edge node [above] {$$}(E);
\end{tikzpicture}
\]
together with two 2-morphisms
\[
\begin{tikzpicture}[scale=1.5]
\node (D) at (-4,0.5) {$S$};
\node (E) at (-3,0.5) {$(S',0_{S'})$};
\node (F) at (-2,0.5) {$S'$};
\node (A) at (-4,-1) {$S'$};
\node (B) at (-2,-1) {$S'$};
\node (G) at (-3,-1) {$(S',0_{S'})$};
\node (D') at (-1,0.5) {$S$};
\node (E') at (0,0.5) {$(S,0_S)$};
\node (F') at (1,0.5) {$S$};
\node (A') at (-1,-1) {$S$};
\node (B') at (1,-1) {$S'$};
\node (G') at (0,-1) {$(S',0_{S'})$};
\path[->,font=\scriptsize,>=angle 90]
(D) edge node [above] {$f$}(E)
(F) edge node {$$}(E)
(D) edge node [left]{$f$}(A)
(F) edge node [right]{$$}(B)
(A) edge node {$$} (G)
(B) edge node {$$} (G)
(E) edge node[left] {$$}(G)
(D') edge node {$$} (E')
(F') edge node {$$} (E')
(D') edge node [left]{$$}(A')
(F') edge node [right]{$f$}(B')
(A') edge node [above] {$f$} (G')
(B') edge node {$f$} (G')
(E') edge node[left] {$f$}(G');
\end{tikzpicture}
\]
such that vertical composition gives
\[
\begin{tikzpicture}[scale=1.5]
\node (D') at (-1,-0.25) {$S$};
\node (E') at (0,-0.25) {$(S,0_S)$};
\node (F') at (1,-0.25) {$S$};
\node (A') at (-1,-1.5) {$S$};
\node (B') at (1,-1.5) {$S'$};
\node (G') at (0,-1.5) {$(S',0_{S'})$};
\node (A) at (-1,-2.75) {$S'$};
\node (B) at (0,-2.75) {$(S',0_{S'})$};
\node (C) at (1,-2.75) {$S'$};
\node (E) at (2,-1.5) {$=$};
\node (D) at (3,-0.5) {$S$};
\node (F) at (3,-2.5) {$S'$};
\node (G) at (4,-0.5) {$(S,0_S)$};
\node (H) at (5,-0.5) {$S$};
\node (I) at (4,-2.5) {$(S',0_{S'})$};
\node (J) at (5,-2.5) {$S'$};
\path[->,font=\scriptsize,>=angle 90]
(D') edge node {$$}(E')
(F') edge node {$$}(E')
(D') edge node [left]{$$}(A')
(F') edge node [right]{$f$}(B')
(A') edge node [above]{$f$} (G')
(B') edge node {$$} (G')
(E') edge node[left] {$f$}(G')
(A') edge node [left] {$f$}(A)
(A) edge node {$$} (B)
(C) edge node {$$} (B)
(B') edge node [right]{$$} (C)
(G') edge node [left]{$$} (B)
(D) edge node {$$} (G)
(D) edge node [left]{$f$} (F)
(F) edge node {$$} (I)
(J) edge node {$$} (I)
(H) edge node [right]{$f$} (J)
(H) edge node {$$} (G)
(G) edge node [left]{$f$} (I);
\end{tikzpicture}
\]
and horizontal composition gives
\[
\begin{tikzpicture}[scale=1.5]
\node (D') at (-1,0.5) {$S$};
\node (E') at (0,.5) {$(S,0_S)$};
\node (F') at (1,0.5) {$S$};
\node (A') at (-1,-1) {$S$};
\node (B') at (1,-1) {$S'$};
\node (G') at (0,-1) {$(S',0_{S'})$};
\node (A) at (2,0.5) {$(S',0_{S'})$};
\node (B) at (2,-1) {$(S',0_{S'})$};
\node (C) at (3,0.5) {$S'$};
\node (D) at (3,-1) {$S'$};
\node (E) at (3.5,-0.25) {$=$};
\node (F) at (4,0.5) {$S$};
\node (G) at (5,0.5) {$(S',0_{S'})$};
\node (H) at (6,0.5) {$S'$};
\node (I) at (4,-1) {$S$};
\node (J) at (5,-1) {$(S',0_{S'})$};
\node (K) at (6,-1) {$S'$};
\path[->,font=\scriptsize,>=angle 90]
(D') edge node {$$}(E')
(F') edge node {$$}(E')
(D') edge node [left]{$$}(A')
(F') edge node [right]{$f$}(B')
(A') edge node [above]{$f$} (G')
(B') edge node {$$} (G')
(E') edge node[left] {$f$}(G')
(F') edge node [above]{$f$} (A)
(C) edge node {$$} (A)
(C) edge node [right]{$$} (D)
(D) edge node {$$} (B)
(B') edge node {$$} (B)
(A) edge node [left] {$$} (B)
(F) edge node [above]{$f$} (G)
(H) edge node [above]{$$} (G)
(G) edge node [left]{$$} (J)
(F) edge node [left]{$$} (I)
(I) edge node [above]{$f$} (J)
(K) edge node {$$} (J)
(H) edge node [right]{$$} (K);
\end{tikzpicture}
\]
A conjoint of $f \maps S \to S'$ is given by the horizontal 1-cell
\[
\begin{tikzpicture}[scale=1.5];
\node (D) at (-4.5,.5) {$S'$};
\node (E) at (-3,.5) {$(S',0_{S'})$};
\node (F) at (-1.5,.5) {$S$};
\path[->,font=\scriptsize,>=angle 90]
(D) edge node [above] {$$}(E)
(F) edge node [above] {$f$}(E);
\end{tikzpicture}
\]
together with two 2-morphisms
\[
\begin{tikzpicture}[scale=1.5]
\node (D) at (-4,0.5) {$S'$};
\node (E) at (-3,0.5) {$(S',0_{S'})$};
\node (F) at (-2,0.5) {$S$};
\node (A) at (-4,-1) {$S'$};
\node (B) at (-2,-1) {$S'$};
\node (G) at (-3,-1) {$(S',0_{S'})$};
\node (D') at (-1,0.5) {$S$};
\node (E') at (0,0.5) {$(S,0_S)$};
\node (F') at (1,0.5) {$S$};
\node (A') at (-1,-1) {$S'$};
\node (B') at (1,-1) {$S$};
\node (G') at (0,-1) {$(S',0_{S'})$};
\path[->,font=\scriptsize,>=angle 90]
(D) edge node [above]{$$}(E)
(F) edge node [above]{$f$}(E)
(D) edge node {$$}(A)
(F) edge node [right]{$f$}(B)
(A) edge node [above]{$$} (G)
(B) edge node [above]{$$} (G)
(E) edge node[left] {$$}(G)
(D') edge node [above]{$$}(E')
(F') edge node [above]{$$}(E')
(D') edge node [left]{$f$}(A')
(F') edge node [right]{$$}(B')
(A') edge node [above]{$$} (G')
(B') edge node [above]{$f$} (G')
(E') edge node[left] {$f$}(G');
\end{tikzpicture}
\]
that satisfy equations analogous to the two above.
\end{proof}

\begin{thm} \label{markbicat}
$\bold{Mark}$ is a symmetric monoidal bicategory.
\end{thm}

\begin{proof}
This follows immediately from Thm.\ \cite{Shulman}: $\MMark$ is an isofibrant symmetric monoidal double category, so we obtain the symmetric monoidal bicategory $\bold{Mark}$ as the horizontal bicategory of $\MMark$.
\end{proof}

We can also obtain a symmetric monoidal bicategory $\bold{LinRel}$ from the symmetric monoidal double category $\LLinRel$ using this fact:

\begin{lem}
The symmetric monoidal double category $\LLinRel$ is isofibrant.
\end{lem}

\begin{proof}
Let $f \maps X \to Y$ be a linear isomorphism between finite-dimensional real vector spaces. Define $\hat{f}$ to be the linear relation given by the linear isomorphism $f$ and define 2-morphisms in $\LLinRel$
\[
\begin{tikzpicture}[scale=1.5]
\node (D) at (-4,0.5) {$X$};
\node (F) at (-2,0.5) {$Y$};
\node (A) at (-4,-1) {$Y$};
\node (B) at (-2,-1) {$Y$};
\node (D') at (0,0.5) {$X$};
\node (F') at (2,0.5) {$X$};
\node (A') at (0,-1) {$X$};
\node (B') at (2,-1) {$Y$};
\node (C) at (-3,-0.25) {$\alpha_{f} \Downarrow$};
\node (C') at (1,-0.25) {${}_{f}\alpha \Downarrow$};
\path[->,font=\scriptsize,>=angle 90]
(D) edge node [above]{$\hat{f}$}(F)
(D) edge node [left]{$f$}(A)
(F) edge node [right]{$1$}(B)
(A) edge node [above]{$1$} (B)
(D') edge node [above]{$1$}(F')
(D') edge node [left]{$1$}(A')
(F') edge node [right]{$f$}(B')
(A') edge node [above]{$\hat{f}$} (B');
\end{tikzpicture}
\]
where $\alpha_{f}$ and ${}_{f} \alpha$, the unique fillers of their frames, are identities. These two 2-morphisms and $\hat{f}$ satisfy the required equations, and the conjoint of $f$ is given by reversing the direction of $\hat{f}$, which is just $f^{-1} \maps Y \to X$. It follows that $\LLinRel$ is isofibrant.
\end{proof}

\begin{thm}
There exists a symmetric monoidal bicategory $\bold{LinRel}$ with
\begin{enumerate}
\item finite-dimensional real vector spaces as objects,
\item linear relations $R \subseteq V \oplus W$ as morphisms from $V$ to $W$,
\item inclusions $R \subseteq S$ between linear relations $R,S \subseteq V \oplus W$ as 2-morphisms.
\end{enumerate}
\end{thm}

\begin{proof}
Apply Shulman's result, Thm.\ \ref{Shulman}, to the isofibrant symmetric monoidal double category $\LLinRel$ to obtain the symmetric monoidal bicategory $\bold{LinRel}$ as the horizontal edge bicategory of $\LLinRel$.
\end{proof}


Thus we have symmetric monoidal bicategories $\bold{Mark}$ and $\bold{LinRel}$, both of which come from discarding the vertical 1-morphisms of the symmetric monoidal double categories $\MMark$ and $\LLinRel$, respectively. Morally, we should be able to do something similar to the symmetric monoidal double functor $\blacksquare \maps \MMark \to \LLinRel$ to obtain a symmetric monoidal functor of bicategories $\blacksquare \maps \bold{Mark} \to \bold{LinRel}$.

\begin{con}
There exists a symmetric monoidal functor $\blacksquare \maps \bold{Mark} \to \bold{LinRel}$ that maps:
\begin{enumerate}
\item any finite set $S$ to the finite-dimensional real vector space $\blacksquare(S)=\R^S \oplus \R^S$,
\item any open Markov process $S \stackrel{i}{\rightarrow} (X,H) \stackrel{o}{\leftarrow} T$ to the linear relation from $\blacksquare(S)$ to $\blacksquare(T)$ given by the linear subspace 
\[ \blacksquare(S \stackrel{i}{\rightarrow} (X,H) \stackrel{o}{\leftarrow} T) = \]
\[  \{ (i^*(v),I,o^*(v),O) : \; H(v) + i_*(I) - o_*(O) = 0 \} \subseteq \R^S \oplus \R^S \oplus \R^T \oplus \R^T ,\]
\item any morphism of open Markov processes 
\[
\begin{tikzpicture}[scale=1.5]
\node (D) at (-4.2,-0.5) {$S$};
\node (A) at (-4.2,-2) {$S$};
\node (B) at (-1.8,-2) {$T$};
\node (E) at (-3,-0.5) {$(X,H)$};
\node (F) at (-1.8,-0.5) {$T$};
\node (G) at (-3,-2) {$(X',H')$};
\path[->,font=\scriptsize,>=angle 90]
(D) edge node[above] {$i_1$}(E)
(A) edge node[above] {$i'_1$} (G)
(B) edge node [above]{$o'_1$} (G)
(F) edge node [above]{$o_1$}(E)
(D) edge node [left]{$1_S$}(A)
(F) edge node [right]{$1_T$}(B)
(E) edge node[left] {$p$}(G);
\end{tikzpicture}
\]
to the inclusion
\[   \blacksquare(X,H) \subseteq \blacksquare(X',H') .\]
\end{enumerate}
\end{con}

\subsection*{Acknowledgements}

We thank Tobias Fritz and Blake Pollard for many helpful conversations, Daniel Cicala
for the commutative cube, and Michael Shulman and other denizens of the $n$-Category
Caf\'e for helping us understand the importance of having cospans with monic legs.  
We thank the anonymous referee for many useful suggestions.
JB also thanks the Centre for Quantum Technologies, where some of this work was done.

\end{document}